\documentclass[lettersize, journal]{IEEEtran}
\usepackage{float}
\usepackage{caption}

\usepackage{color}

\usepackage{cite}
\usepackage{amsmath,amssymb,amsfonts}
\usepackage{amsthm}
\usepackage{graphicx}
\usepackage{textcomp}
\usepackage{xcolor}
\usepackage{svg}
\usepackage{subcaption}

\usepackage{algorithm}
\usepackage{mathtools}

\usepackage{algpseudocode}
\usepackage{epstopdf}
\usepackage{multirow}
\usepackage{pifont}

\usepackage{stfloats}

\newtheorem{theorem}{Theorem}

\newtheorem{lemma}{Lemma}

\newtheorem{corollary}{Corollary}

\allowdisplaybreaks
\makeatletter
\def\ScaleIfNeeded{%
\ifdim\Gin@nat@width>\linewidth \linewidth \else \Gin@nat@width
\fi } \makeatother

\graphicspath{{Fig/}}

\begin{document}

\title{Pinching-Antenna Systems-Enabled Multi-User Communications: Transmission Structures and Beamforming Optimization}
\author{Jingjing Zhao, Haowen Song, Xidong Mu, Kaiquan Cai, Yanbo Zhu, and Yuanwei Liu,~\IEEEmembership{Fellow,~IEEE}
\thanks{J. Zhao, H. Song, K. Cai, Y. Zhu are with the School of Electronics and Information Engineering, Beihang University, 100191, Beijing, China, and also with the State Key Laboratory of CNS/ATM, 100191, Beijing, China. (e-mail:\{jingjingzhao, haowensong, ckq, zyb\}@buaa.edu.cn). 

X. Mu is with the Centre for Wireless Innovation (CWI), Queen's University Belfast, Belfast, BT3 9DT, U.K. (e-mail: x.mu@qub.ac.uk). 

Y. Liu is with the Department of Electrical and Electronic Engineering, the University of Hong Kong, Hong Kong, China (e-mail: yuanwei@hku.hk). }}
\maketitle
\begin{abstract}
Pinching-antenna systems (PASS) represent an innovative advancement in flexible-antenna technologies, aimed at significantly improving wireless communications by ensuring reliable line-of-sight connections and dynamic antenna array reconfigurations. To employ multi-waveguide PASS in multi-user communications, three
practical transmission structures are proposed, namely waveguide multiplexing (WM), waveguide division (WD), and waveguide switching (WS). Based on the proposed structures, the joint baseband signal processing and pinching beamforming design is studied for a general multi-group multicast communication system, with the unicast communication encompassed as a special case. A max-min fairness (MMF) problem is formulated for each proposed transmission structure, subject to the maximum transmit power constraint. For WM, to solve the highly-coupled and non-convex MMF problem with complex exponential and fractional expressions, a penalty dual decomposition (PDD)-based algorithm is invoked for obtaining locally optimal solutions. Specifically, the augmented Lagrangian relaxation is first applied to alleviate the stringent coupling constraints, which is followed by the block decomposition over the resulting augmented Lagrangian function. Then, the proposed PDD-based algorithm is extended to solve the MMF problem for both WD and WS. Furthermore, a low-complexity algorithm is proposed for the unicast case employing the WS structure, by simultaneously aligning the signal phases and minimizing the large-scale path loss at each user. 
Finally, numerical results reveal that: 1) the MMF performance is significantly improved by employing the PASS compared to conventional fixed-position antenna systems; 2) WS and WM are suitable for unicast and multicast communications, respectively; 3) the performance gap between WD and WM can be significantly alleviated when the users are geographically isolated.
\end{abstract}
\section{Introduction}
The proliferation of data-intensive applications and the growing demand for ubiquitous high-speed connectivity have placed unprecedented pressure on the physical layer design of wireless communication systems. Among the various innovations in the air-interface design, multiple-input multiple-output (MIMO)~\cite{6736761, 11159291} stands out as a pivotal technology that has shaped the system design over the past decades. By employing antenna arrays at the transmitter and receiver, MIMO systems are capable of boosting signal quality via directional transmission, alleviating the effects of multi-path propagation, and enabling the concurrent transmission of multiple data streams within the same frequency band~\cite{6375940}. However, the fixed and discrete antenna deployment restricts the diversity and spatial multiplexing gain of conventional MIMO systems, as the channel variation in the continuous spatial field is not fully exploited. 

To address the demanding capacity needs of upcoming sixth-generation (6G) networks, flexible and reconfigurable antenna technologies have gained considerable focus as an advanced evolution of MIMO. Among the most notable advances are reconfigurable intelligent surfaces (RISs)~\cite{marco,8910627}, dynamic metasurface antennas~\cite{9324910}, holographic MIMO surfaces~\cite{9136592}, and recently proposed movable/fluid antenna systems~\cite{9264694,10318061}. While these technologies have demonstrated spectrum efficiency gains by dynamically adjusting the propagation environment or antenna geometry, they remain fundamentally constrained by issues, such as large-scale propagation impairments, i.e., free-space path loss and line-of-sight (LoS) blockage, and limited flexibility in scaling up/down the number of antennas.
To overcome these challenges, the pinching-antenna system (PASS) has emerged as a novel antenna architecture that offers a new degree of reconfigurability and spatial adaptability~\cite{ding,11169486}. Initially prototyped by NTT DOCOMO~\cite{docomo}, PASS employs dielectric waveguides with low propagation loss as the transmission medium and utilizes pinching antennas (PAs)—small dielectric particles attached along the waveguide—as radiating elements. Through precise control of the activation positions of PAs, PASS enables the pinching beamforming, a new paradigm that jointly optimizes the large-scale path loss and the phase of transmitted signals. 
More specifically, compared to existing flexible-antenna systems, the main profits of PASS rely on the following three aspects. \textit{Firstly}, \textcolor{blue}{since long waveguides can be deployed in environments such as tunnels and airports in a manner similar to leaky coaxial cables (LCXs), PAs can be activatd at any point along the waveguide to enable the “last-meter” communication, which facilitates establishing favorable propagation conditions and maintaining stable LoS links in high-frequency bands.} \textit{Secondly}, the ``plug-and-play" modular structure allows scalable and cost-efficient antenna deployment. \textcolor{blue}{It is worthy noting that, although distributed MIMO can also improve communication quality  by moving antennas close to users, it requires dedicated radio-frequency chain for each antenna  and leads to higher hardware cost.} \textit{Thirdly}, the capability of simultaneously configuring baseband and pinching beamforming enables enhanced spectrum and energy efficiency performance.

The aforementioned benefits have motivated early-stage research endeavors focused on the employment of PASS in wireless communications. The authors of~\cite{ding} first presented performance analysis results of the PASS under various cases. By exploiting the PAs capability of reconfiguring wireless channels, it was verified that PASS can achieve superior performance compared to conventional fixed-position antenna systems. Furthermore, a comprehensive analytical framework was developed in~\cite{10976621}, where closed-form expressions for the outage probability and average rate were derived. The authors also characterized the optimal PAs positions for performance maximization under the consideration of waveguide attenuation in realistic conditions.
From the perspective of the pinching beamforming optimization, the authors of~\cite{yanqing} tackled with the PAs activation problem for maximizing the pinching beamforming gain under the simple single-waveguide and single-user scenario. Closed-form PAs activation positions were derived for simultaneously reducing the large-scale path loss and aligning the phases of signals from multiple PAs. Since each waveguide can only convey one data stream, the non-orthogonal multiple access (NOMA) assisted PASS scheme was further proposed in~\cite{kaidi} for the case of multiple users served by a single waveguide, where the PAs could only be activated at pre-configured positions. Moreover, the authors of~\cite{10909665} investigated the uplink performance of the multi-user PASS, where users transmit data to the base station (BS) equipped with a single waveguide in the orthogonal multiple access (OMA) manner. A minimum achievable data rate maximization problem was formulated, which was efficiently solved by realizing the phase alignment of signals from PAs and minimizing path losses.

\subsection{Motivations and Contributions}
By substituting conventional antenna elements with waveguides and radiating signals via activated PAs, the pinching beamforming in PASS can unlock additional spatial degrees of freedom (DoFs). However, existing literatures predominantly focus on the single-waveguide design for PASS. To support multi-user/multi-data stream transmission, multi-waveguide PASS should be employed, where each waveguide is connected with one RF chain to exploit the baseband signal processing capability. \textcolor{blue}{The authors in~\cite{shan2025multigroup} investigated the piniching beamforming in multi-waveguide PASS for supporting muti-group multicast communications. However, the systematical tansmission structure design was not discussed.} Although the baseband signal processing provides more DoFs, it leads to potential high computational complexity. To strike an appropriate trade-off between system performance and implementation complexity, it is imperative to devise customized transmission structures tailored for multi-waveguide PASS. This provides the primary motivation of this work.

To address the above issues, in this paper, we propose three basic transmission structures for PASS and study the joint baseband signal processing and pinching beamforming design for a general multi-group multicast communication system. The main contributions of this work are summarized as follows.
\begin{itemize}
    \item We propose three basic transmission structures for PASS, namely, \textit{waveguide multiplexing (WM)}, \textit{waveguide division (WD)}, and \textit{waveguide switching (WS)}, and discuss the corresponding benefits and drawbacks. 
    \item Building upon the proposed transmission structures, we investigate the joint baseband processing and beamforming design for the multi-group multicast communication system, with the unicast communication acting as a special case. A max-min fairness (MMF) optimization problem is formulated for each of the proposed structure, subject to the maximum transmit power constraint.  
    \item For WM, we invoke the penalty dual decomposition (PDD) method for solving the highly-coupled and non-convex MMF problem incorporating complex exponential and fractional terms. In particular, we employ the augmented Lagrangian relaxation to alleviate the stringent coupling constraints. Then, the alternating optimization (AO) is applied for iteratively solving the decomposed subproblems. Furthermore, we extend the proposed PDD algorithm to solve the MMF problems for both WD and WS. For the unicast case under WS, we also introduce a low-complexity algorithm to concurrently align the received signal phases and minimize the large-scale path loss for each individual user.

    \item Numerical results unveil that 1) the MMF performance is notably enhanced by the integration of PASS, outperforming conventional fixed-position MIMO system under both fully-digital and hybrid beamforming structures; 2) WS shows its superiority for unicast communications, while WM is preferable for multicast communications; and 3) the performance disparity between WM and WD can be substantially reduced when users are geographically dispersed.
    
\end{itemize}

\subsection{Organization and Notations}
The rest of this paper is structured as follows. Section II proposes three practical transmission structures to operate multi-waveguide PASS for multi-user/data transmission, based on which the signal radiation models are accordingly introduced. Section III presents the system model and the joint baseband processing and pinching beamforming problem formulations for the general multi-group multicast communication, with the aim of maximizing the minimum achievable rate across all users. Section IV proposes efficient PDD-based algorithms for solving the problems formulated for each proposed transmission structure. Section V provides numerical results to verify the superiority of PASS compared to state-of-the-art baselines. Finally, Section VI concludes the paper.  

$\textit {Notations}$: Scalars, vectors, and matrices are denoted by italic letters, bold-face lower-case, and bold-face upper-case, respectively. $\mathbb{C}^{N\times M}$ denotes the set of $N\times M$ complex-valued matrices. Superscripts $(\cdot)^*, (\cdot)^T, (\cdot)^H$, and $(\cdot)^{-1}$ denote the conjugate, transpose, conjugate transpose, and inversion operators, respectively. $|\cdot|$ and  $\left\|\cdot\right\|$ denote the  determinant and Euclidean norm of a matrix, respectively. $\text{Tr}\left(\cdot\right)$, $\left\|\cdot\right\|_F$, and $\text{vec}\left(\cdot\right)$ denote the trace, Frobenius norm, and vectorization of a matrix,  respectively. $[\cdot]_{m,n}$ denotes the $(m,n)$-th element of a matrix. $\mathbf{1}_{{N}}$ denotes the all-one row vector with length $N$. 
$\mathbb{E}$ denotes the expectation operator. $\circ$ denotes the Hadamard multiplication. All random variables are assumed to be \textit{zero} mean. 
\section{Practical Transmission Structures for Multi-Waveguide PASS}
Consider a downlink PASS consisting of $K$ dielectric waveguides, each fed by a dedicated RF chain and incorporating $N$ PAs. Denote the set of waveguides, the set of PAs on the $k$-th waveguide, and the set of all PAs by $\mathcal{K}_{\text{WG}} = \left\{1,\dots, K\right\}$, $\mathcal{N}_k=\left\{1,\dots, N\right\}$, and $\mathcal{M}=\left\{1, \dots, M\right\}$, respectively, with
$M=K\times N$. Let $\mathbf{s}=\left[s_1,\dots, s_K\right]\in\mathbb{C}^{1\times K}$ denote the vector of transmitted data streams, for which the size equals to the number of waveguides. Note that the transmitted signals are assumed to be normalized with $\mathbb{E}\left[\mathbf{s}\mathbf{s}^{H}\right]=\mathbf{I}_{K}$. 
All waveguides are assumed to be aligned parallel to the $x$-axis at a height of $d$. Accordingly, the positions of the feed point and the $n$-th antenna on the $k$-th waveguide are denoted by $\bar{\boldsymbol{\psi}}^k_{\text{p}}=[0,y_{\text{p}}^k,d]$ and $\boldsymbol{\psi}_{\text{p}}^{k,n} = \left[x_{\text{p}}^{k,n}, y_{\text{p}}^{k}, d\right]$, respectively.  Further denote the set of $x$-axis positions of PAs on the $k$-th waveguide and that of $x$-axis positions of PAs over all waveguides by $\mathbf{x}_{\text{p}}^{k} = \left[x_{\text{p}}^{k,1}, ..., x_{\text{p}}^{k,N}\right]$ and $\mathbf{X} = \left[\left(\mathbf{x}_{\text{p}}^{1}\right)^T, ..., \left(\mathbf{x}_{\text{p}}^{K}\right)^T\right]^T\in\mathbb{R}^{K\times N}$, respectively. Assume that PAs on each waveguide are placed in a successive order, i.e., $x_{\text{p}}^{k,{n+1}}>x_{\text{p}}^{k,{n}}, \forall 1\leq n< N, \forall k$, and the maximum deployment range of PAs is $L$. \textcolor{blue}{For supporting flexible adjustment of PAs positions, each PA is mounted on a sliding track over the waveguide, which requires a linear actuator, such as a servo or stepper motor, along with a position sensor and real-time control logic~\cite{liu2025pinching}.} The signal propagation response $\mathbf{g}\left(\mathbf{x}_{\text{p}}^{k}\right)\in\mathbb{C}^{N\times 1}$ from the feed point to PAs over the $k$-th waveguide is given by
\begin{align}
\label{eq:waveguide-channel}
    \mathbf{g}\left(\mathbf{x}_{\text{p}}^{k}\right) & = \frac{1}{\sqrt{N}}\left[e^{-j\frac{2\pi \left\|\boldsymbol{\psi}_{\text{p}}^{k,1}-\bar{\boldsymbol{\psi}}_{\text{p}}^k\right\|}{\lambda_{\text{g}}}},..., e^{-j\frac{2\pi \left\|\boldsymbol{\psi}_{\text{p}}^{k,N}-\bar{\boldsymbol{\psi}}_{\text{p}}^k\right\|}{\lambda_{\text{g}}}}\right]^T\nonumber\\
    &=\frac{1}{\sqrt{N}}\left[e^{-j\frac{2\pi x_{\text{p}}^{k,1}}{\lambda_{\text{g}}}},..., e^{-j\frac{2\pi x_{\text{p}}^{k,N}}{\lambda_{\text{g}}}}\right]^T,
\end{align}
where $\lambda_{\text{g}}=\frac{\lambda}{n_{\text{eff}}}$ denotes the guided wavelength with $\lambda$ and $n_{\text{eff}}$ representing the signal wavelength in the free-space and the effective refractive index of a dielectric waveguide~\cite{microwave}, respectively. We assume that each PA radiates an equal proportion of the power allocated to the corresponding waveguide. Therefore, the coefficient $\frac{1}{\sqrt{N}}$ is multiplied to the propagation response. \textcolor{blue}{Note that the waveguide propagation loss is omitted in~\eqref{eq:waveguide-channel}, which is negligible compared to the free-space path loss as demonstrated in~\cite{kaidi, xu2025pinching}.}

In the following, we introduce three practical transmission structures for PASS and present the corresponding signal radiation models.

\begin{figure}
    \centering
    \begin{subfigure}{\linewidth}
        \centering
        \includegraphics[width=0.9\linewidth]{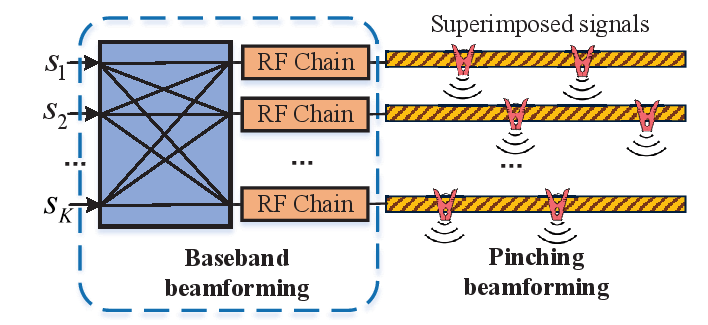}
        \caption{Waveguide multiplexing.}
        \label{fig:WM}
    \end{subfigure}
    \begin{subfigure}{\linewidth}
        \centering
        \includegraphics[width=0.9\linewidth]{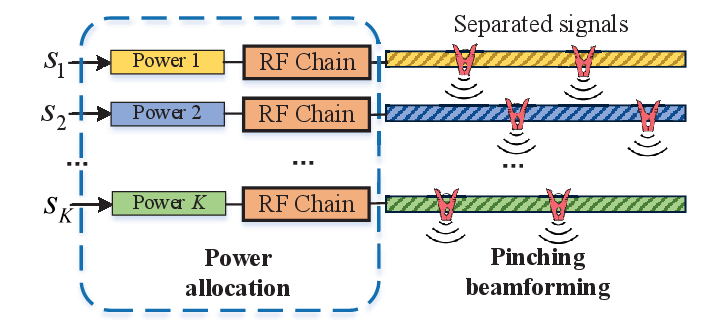}
        \caption{Waveguide division.}
        \label{fig:WD}
    \end{subfigure}
    \begin{subfigure}{\linewidth}
        \centering
        \includegraphics[width=0.9\linewidth]{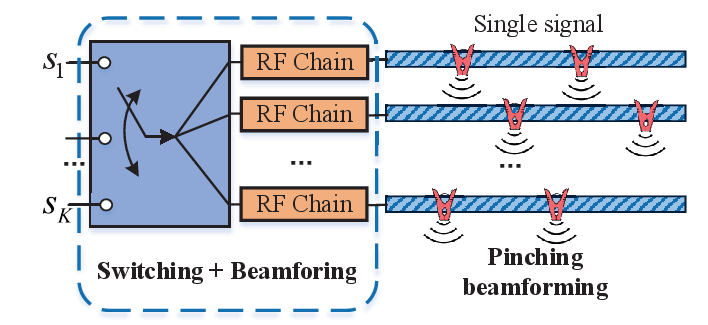}
        \caption{Waveguide switching.}
        \label{fig:TS}
    \end{subfigure}
    \caption{Proposed transmission structures for multi-waveguide PASS.}
    \label{fig:ThreeStructuresVertical}
\end{figure}
\subsection{Waveguide Multiplexing}
As shown in Fig. 1(a), for the WM structure, data streams are multiplexed at the baseband through the conventional digital beamforming, and then fed into the waveguides via the corresponding RF chains. Denote the baseband beamforming vector for the data stream $s_k$ by $\mathbf{w}_k\in\mathbb{C}^{K\times 1}$, then the signal radiated by the PASS to the free space is given by 
\begin{equation}
    \label{eq:WM-radiation}
    \tilde{\mathbf{s}}^{\text{WM}}= \mathbf{G}\left(\mathbf{X}\right)\sum_{k=1}
    ^K\mathbf{w}_k{s}_k\in\mathbb{C}^{M\times 1},
\end{equation}
where $\mathbf{G}\left(\mathbf{X}\right)\in\mathbb{C}^{M\times K}$ denotes the propagation response from feed points to PAs over all waveguides, given by
\begin{equation}
    \mathbf{G}\left(\mathbf{X}\right)=
    \begin{bmatrix}
    \mathbf{g}\left(\mathbf{x}_{\text{p}}^1 \right) & \mathbf{0} & \dots & \mathbf{0}\\

    \mathbf{0} & \mathbf{g}\left(\mathbf{x}_{\text{p}}^2 \right) & \dots & \mathbf{0}\\

    \vdots & \vdots & \ddots & \vdots\\

    \mathbf{0} &\mathbf{0} & \dots & \mathbf{g}\left(\mathbf{x}_{\text{p}}^K \right)
    \end{bmatrix}.
\end{equation}
\textcolor{blue}{For WM, since both the baseband and pinching beamforming are enabled, high degree of flexibility for the spatial-domain multiplexing is facilitated. However, the large number of variables and the tight coupling between the baseband and pinching beamforming parameters lead to complex system design.} 
\subsection{Waveguide Division}
As shown in Fig. 1(b), for the WD structure, each data stream is conveyed by one dedicated waveguide, where the baseband essentially carries out the power allocation. Without loss of generality, we assume that the data stream $s_k$ is conveyed by the $k$-th waveguide. 
Let $\mathbf{p}=\left[p_1, ..., p_{K}\right]$ denote the power allocation vector to the total $K$ data streams, then the signal radiated by the PASS to the free space is given by
\begin{equation}
    \label{eq:WD-radiation}
    \tilde{\mathbf{s}}^{\text{WD}}= \mathbf{G}\left(\mathbf{X}\right)
    \begin{bmatrix}
        \sqrt{p_1}s_1\\
        \vdots \\
        \sqrt{p_K}s_K\\
        
    \end{bmatrix}\in\mathbb{C}^{M\times 1}.
\end{equation}
\textcolor{blue}{We note that WD generally cannot achieve the same multiplexing gain as WM, since the baseband beamforming is deactivated. Nevertheless, WD is still appealing in practice, especially when the waveguides are deployed in geographically separate areas and the baseband multiplexing becomes inefficient. As only power allocation needs to be carried out at the baseband, the design complexity for WD is much reduced compared to that for WM.}  
\subsection{Waveguide Switching}
\textcolor{blue}{Different from WM and WD, the WS exploits the time domain and periodically serves different users via individual baseband and pinching beamforming in orthogonal time slots, as shown in Fig. 1(c).} Hence, when the data stream $s_k$ is conveyed, the signal radiated by the PASS to the free space is given by
\begin{equation}
    \label{eq:TS-radiation}
    \tilde{\mathbf{s}}^{\text{WS}}\left(t\right)= \mathbf{G}\left(\mathbf{X}\right)\mathbf{w}_k(t){s}_k(t)\in\mathbb{C}^{M\times 1}.
\end{equation}
\textcolor{blue}{By exploiting the extra time domain, WS can effectively avoid the multi-user interference. Moreover, the design of the baseband processing and the pinching beamforming for accommodating different users is not coupled, and thus, making it easier to handle. However, the periodical time-domain switching imposes stringent time synchronization requirements, which necessitates higher implementation complexity.}

In Table~I, we summarize the characteristics of the WM, WD, and WS transmission structures.

\begin{table*}[t] 
\centering
\renewcommand{\arraystretch}{1.3} 
\caption{Characteristics of the three proposed transmission architectures.}
\begin{tabular}{|c|c|c|c|}
\hline
\textbf{} & \textbf{Waveguide Multiplexing} & \textbf{Waveguide Division} & \textbf{Waveguide Switching} \\
\hline
Baseband beamforming & \textcolor{blue}{\checkmark} & \textcolor{blue}{\ding{55}} & \textcolor{blue}{\checkmark} \\
\hline
Pinching beamforming 
& \textcolor{blue}{\checkmark} & \textcolor{blue}{\checkmark} & \textcolor{blue}{\checkmark} \\
\hline
Time allocation & \textcolor{blue}{\ding{55}} & \textcolor{blue}{\ding{55}} & \textcolor{blue}{\checkmark} \\
\hline
Multi-user interference & \textcolor{blue}{\checkmark} & \textcolor{blue}{\checkmark} & \textcolor{blue}{\ding{55}} \\
\hline
\end{tabular}
\end{table*}

\section{System Model and Problem Formulation}
Based on the proposed transmission structures, in this section, we present the system model of the multi-group multicast communication in PASS and formulate the MMF optimization problem.   

\subsection{System Model}
Consider there exist $K$ multicast groups, with $G$ users in each group sharing the same information data. \textcolor{blue}{Let $\mathcal{K}_{\text{MC}}=\left\{1, \dots, K\right\}$ and $\mathcal{G}_{k}=\left\{1, \dots, G\right\}$ denote the set of multicast groups and that of users in the $k$-th multicast group, respectively. Denote the location of the $g$-th user in the $k$-th multicast group, i.e., the user $(k,g)$, by $\boldsymbol{\psi}_{\text{u}}^{k,g}=\left[x_{\text{u}}^{k,g},y_{\text{u}}^{k,g},0\right]$.} Then, the wireless channel between PAs on the $i$-th waveguide and the user $(k,g)$ is given by
\begin{align}
    & \mathbf{h}_{k,g}^{i}\left(\mathbf{x}_{\text{p}}^{i}\right)= \left[\frac{\eta e^{-j\frac{2\pi}{\lambda}\left\|\boldsymbol{\psi}^{k,g}_{\text{u}}-\boldsymbol{\psi}^{i,1}_{\text{p}}\right\|}}{\left\|\boldsymbol{\psi}^{k,g}_{\text{u}}-\boldsymbol{\psi}^{i,1}_{\text{p}}\right\|}, ..., \frac{\eta e^{-j\frac{2\pi}{\lambda}\left\|\boldsymbol{\psi}^{k,g}_{\text{u}}-\boldsymbol{\psi}_{\text{p}}^{i,N}\right\|}}{\left\|\boldsymbol{\psi}^{k,g}_{\text{u}}-\boldsymbol{\psi}_{\text{p}}^{i,N}\right\|}\right],
    \label{eq:h-channel}
\end{align}
where $\eta=\frac{\lambda}{4\pi}$ represents the channel gain at the reference distance of $1$~m, and $\left\|\boldsymbol{\psi}^{k,g}_{\text{u}}-\boldsymbol{\psi}_{\text{p}}^{i,n}\right\|$ is given by
\begin{equation}
   \left\|\boldsymbol{\psi}^{k,g}_{\text{u}}-\boldsymbol{\psi}_{\text{p}}^{i,n}\right\|=\sqrt{\left(x_{\text{u}}^{k,g}-x_{\text{p}}^{i,n}\right)^2+\left(y_{\text{u}}^{k,g}-y_{\text{p}}^{i}\right)^2+d^2}. 
\end{equation}
Further denote $\mathbf{h}_{k,g}\left(\mathbf{X}\right)=\left[\mathbf{h}^{1}_{k,g}\left(\mathbf{x}_{\text{p}}^1\right), \dots, \mathbf{h}^{K}_{k,g}\left(\mathbf{x}_{\text{p}}^K\right)\right]\in\mathbb{C}^{1\times M}$ as the channel vector from PAs over all waveguides to the $(k,g)$ user\textcolor{blue}{\footnote{\textcolor{blue}{Compared to hybrid MIMO, due to the deterministic nature of in-waveguide channel in PASS, it is intractable to collect independent pilot measurements for estimating the channel to each PA. This makes it difficult for the high-dimensional channel estimation with low-dimensional observations. To solve this issue, the authors of~\cite{11018390} introduced the antenna switching matrix-based least square and linear minimum mean square error estimators, which selected subsets of PAs in different time slots to realize the signal separation. In this work, we assume that the CSI is perfectly known at the BS with the method proposed in~\cite{11018390}, so as to explore the maximum performance brought by PASS.}}}. In the following, we discuss the signal model for the transmission structures WM, WD, and WS, separately.
\subsubsection{WM}
Based on the signal radiation model given in Eq.~\eqref{eq:WM-radiation}, the signal received by the user $(k,g)$ for WM can be given by\textcolor{blue}{\footnote{\textcolor{blue}{In this paper, we assume that the delay spread across PAs is much less than the symbol period given the narrowband setup. For the broadband case where the channel exhibits frequency selectivity, the orthogonal frequency division modulation (OFDM) design proposed in~\cite{11180028} can be adopted to avoid inter-symbol interference.}}}
\begin{align}
\label{eq:signal-WM}
        y_{k,g}^{\text{WM}} = &\underbrace{\overbrace{\mathbf{h}_{k,g}\left(\mathbf{X}\right)\mathbf{G}\left(\mathbf{X}\right)}^{\text{all-waveguide pinching beamforming}}\mathbf{w}_ks_k}_{\text{desired signal}} \nonumber\\
    & + \underbrace{\sum_{k'\neq k}^K \overbrace{\mathbf{h}_{k,g}\left(\mathbf{X}\right)\mathbf{G}\left(\mathbf{X}\right)}^{\text{all-waveguide pinching beamforming}}\mathbf{w}_{k'}s_{k'}}_{\text{interfering signal}} + n_{k,g},
\end{align}
where $n_{k,g}\sim \mathcal{CN}(0, \sigma_{k,g}^2)$ represents the additive white Gaussian noise (AWGN) at the $(k,g)$ user, with $\sigma_{k,g}^2$ denoting the noise power. Observing~\eqref{eq:signal-WM}, one can find that the pinching beamforming for each data stream is performed jointly by all waveguides. 
Hence, the achievable rate of the user $(k,g)$ for WM is given by
\begin{align}
\label{eq:rate-WM}
    R_{\text{MC},k,g}^{\text{WM}} & =\log_2\left(1+\frac{\left|\mathbf{h}_{k,g}\left(\mathbf{X}\right)\mathbf{G}\left(\mathbf{X}\right)\mathbf{w}_k\right|^2}{\sum\limits_{k'\neq k}^K \left|{\mathbf{h}_{k,g}\left(\mathbf{X}\right)\mathbf{G}\left(\mathbf{X}\right)}\mathbf{w}_{k'}\right|^2+\sigma_k^2}\right)\nonumber\\
    & =\log_2\left(1+\frac{\left|\sum\limits_{i=1}^K\sum\limits_{n=1}^Nw_{k,i}\frac{\eta e^{-j\phi_{k,g}^{i,n}}}{\left\|\boldsymbol{\psi}_{\text{u}}^{k,g}-\boldsymbol{\psi}_{\text{p}}^{i,n}\right\|}\right|^2}{I^{\text{WM}}_{k,g}+N\sigma_{k,g}^2}\right),
\end{align}
where $w_{k,i}$ represents the $i$-th element of $\mathbf{w}_k$, and $\phi_{k,g}^{i,n}$ is given by
\begin{equation}
    \phi_{k,g}^{i,n} \triangleq \frac{2\pi}{\lambda}\left\|\boldsymbol{\psi}_{\text{u}}^{k,g}-\boldsymbol{\psi}_{\text{p}}^{i,n}\right\| + \frac{2\pi}{\lambda_g}{x}_{\text{p}}^{i,n}.
\end{equation}
Still referring to~\eqref{eq:rate-WM},
$I^{\text{WM}}_{k}$ denotes the interference signal power, given by
\begin{equation}
    I^{\text{WM}}_{k,g}=\sum\limits_{k'\neq k}^K\left|\sum\limits_{i=1}^K\sum\limits_{n=1}^Nw_{k',i}\frac{\eta e^{-j\phi_{k,g}^{i,n}}}{\left\|\boldsymbol{\psi}_{\text{u}}^{k,g}-\boldsymbol{\psi}_{\text{p}}^{i,n}\right\|}\right|^2.
\end{equation}
\subsubsection{WD}
Different from WM, the pinching beamforming for each data stream is performed by its designated waveguide under the WD structure. As such, based on the signal radiation model given in Eq.~\eqref{eq:WD-radiation}, the signal received by the $(k,g)$ user for WD is given by
\begin{align}
\label{eq:signal-WD}
        y_{k,g}^{\text{WD}} & = \underbrace{\overbrace{\mathbf{h}_{k,g}^{k}\left(\mathbf{x}_{\text{p}}^{k}\right)\mathbf{g}\left(\mathbf{x}_{\text{p}}^{k}\right)}^{\text{per-waveguide pinching beamforming}}\sqrt{{p_k}}s_k}_{\text{desired signal}} \nonumber\\
    & + \underbrace{\sum_{k'\neq k}^K \overbrace{\mathbf{h}_{k,g}^{k'}\left(\mathbf{x}_{\text{p}}^{k'}\right)\mathbf{g}\left(\mathbf{x}_{\text{p}}^{k'}\right)}^{\text{per-waveguide pinching beamforming}}\sqrt{p_{k'}}s_{k'}}_{\text{interfering signal}} + n_{k,g}.
\end{align}
Accordingly, the achievable rate of the user $(k,g)$ is given by 
\begin{align}
    & R_{\text{MC},k}^{\text{WD}}\nonumber\\
    &= \log _2\left(1+\frac{p_k\left|\left(\mathbf{h}_{k,g}^{k}\left(\mathbf{x}_{\text{p}}^{k}\right)\right)^T\mathbf{g}\left(\mathbf{x}_{\text{p}}^{k}\right)\right|^2}{\sum\limits_{k'\neq k}^{K}p_{k'}\left|\left(\mathbf{h}^{k'}_{k,g}\left(\mathbf{x}_{\text{p}}^{k'}\right)\right)^T\mathbf{g}\left(\mathbf{x}_{\text{p}}^{k'}\right)\right|^2+\sigma_{k,g}^2}\right)\nonumber\\
    &= \log_2\left(1+\frac{p_k\left|\sum\limits_{n=1}^N\frac{\eta e^{-j\phi_{k,g}^{k,n}}}{\left\|\boldsymbol{\psi}_{\text{u}}^{k,g}-\boldsymbol{\psi}_{\text{p}}^{k,n}\right\|}\right|^2}{\sum\limits_{k'\neq k}^Kp_{k'}\left|\sum\limits_{n=1}^N\frac{\eta e^{-j\phi_{k,g}^{k',n}}}{\left\|\boldsymbol{\psi}_{\text{u}}^{k,g}-\boldsymbol{\psi}_{\text{p}}^{k',n}\right\|}\right|^2+N\sigma_{k,g}^2}\right).
    \label{eq:rate-WD}
\end{align}
\subsubsection{WS}
For WS, only the scheduled signal is transmitted to its corresponding user in each time slot. Therefore, when the signal $s_k$ is conveyed, the signal received by the user $(k,g)$ is given by
\begin{align}
\label{eq:signal-TS}
        y_{k,g}^{\text{WS}} = &\overbrace{\mathbf{h}_{k,g}\left(\mathbf{X}_k\right)\mathbf{G}\left(\mathbf{X}_k\right)}^{\text{all-waveguide pinching beamforming}}\mathbf{w}_ks_k + n_{k,g},
\end{align}
where $\mathbf{X}_k$ denotes the PAs position matrix for serving group $k$. 
From Eq.~\eqref{eq:signal-TS}, one can observe that the inter-group interference does not exist in the received signal.
Let $\boldsymbol{\lambda}=\left\{\lambda_1, \dots, \lambda_K\right\}$ denote the vector of normalized transmission time allocated to total $K$ data streams, where $0 \leq \lambda_k\leq 1$, $\sum_{k=1}^K\lambda_k=1$. 
Then, the achievable rate of the user $(k,g)$ for the WS is given by

\begin{align}
\label{eq:rate-TS}
    R_{\text{MC},k,g}^{\text{WS}}&=\lambda_k\log_2\left(1+\frac{\left|\mathbf{h}_{k,g}\left(\mathbf{X}_k\right)\mathbf{G}\left(\mathbf{X}_k\right)\mathbf{w}_k\right|^2}{\sigma_{k,g}^2}\right)\nonumber\\
    & =\lambda_k\log_2\left(1+\frac{\left|\sum\limits_{i=1}^K\sum\limits_{n=1}^Nw_{k,i}\frac{\eta e^{-j\phi_{k,g}^{i,n}}}{\left\|\boldsymbol{\psi}_{\text{u}}^{k,g}-\boldsymbol{\psi}_{\text{p}}^{i,n}\right\|}\right|^2}{N\sigma_{k,g}^2}\right),
\end{align}
where the multiplier $\lambda_k$ is given due to the fact that the BS sends information to the group $k$ with $\mathbf{w}_k$ by employing $\lambda_k$ fraction of the total transmission time. 

\subsection{Problem Formulation}
\textcolor{blue}{Owing to the mutual dependence among users whose performance is jointly influenced by the shared configuration of PAs positions, it is crucial to strike an appropriate trade-off that ensures users fairness while optimizing overall system performance. To this end, we focus on maximizing the minimum achievable data rate across all users, facilitating the MMF problem.} Specifically, for multicast communications, let $\mathbf{W}=\left\{\mathbf{w}_k\right\}_{\forall k}$, then the optimization problem for WM, WD, and WS can be formulated as
\begin{subequations}
\label{eq:multicast-optimization-problem}
\begin{equation}
\label{eq:multicast-objective-function}
\mathcal{P}_{\text{MC}}(0): \max_{\left\{\mathbf{X},\mathbf{W},\mathbf{p},\boldsymbol{\lambda}\right\}}\min_{k\in\mathcal{K}_{\text{MC}}}\min_{g\in\mathcal{G}_{\text{k}}}R_{\text{MC},k,g}^{\text{Y}},
\end{equation}
\begin{equation}
    \label{eq:multicast-maximum-power}
{\rm{s.t.}} \ \  P_{\text{max}}\geq \left\{
    \begin{array}{c}
    \sum_{k=1}^K\left\|\mathbf{w}_k\right\|^2, \text{if Y=WM},\\
    \sum_{k=1}^Kp_k, \text{if Y=WD},\\
   \left\|\mathbf{w}_k\right\|^2, \forall k, \text{if Y=WS},
    \end{array}
    \right.
\end{equation}
\begin{equation}
\label{eq:multicast-position_constraint1}
0\leq x_{\mathfrak{p}}^{k,n}\leq L,\forall n,k,
\end{equation}
\begin{equation}
\label{eq:multicast-position_constraint2}
x_{\mathrm{p}}^{k,n+1}-x_{\mathrm{p}}^{k,n}\geq\Delta,\forall n,k,
\end{equation}
\begin{equation}
\label{eq:multicast-time-allocation-constraint}
    0\leq \lambda_k\leq 1, \sum_{k=1}^K\lambda_k=1, \forall k,
\end{equation}
\end{subequations}
where $\text{Y}\in$ $\left\{\text{WM, WD, WS}\right\}$ indicates the applied transmission structure. The baseband beamforming variable $\mathbf{W}$, the power allocation vector $\mathbf{p}$, and the time allocation vector $\boldsymbol{\lambda}$ are only valid when the WM/WS, the WD, and the WS is employed, respectively. 
Constraint~\eqref{eq:multicast-maximum-power} ensures the BS transmit power does not exceed the threshold $P_{\text{max}}$. Constraint~\eqref{eq:multicast-position_constraint1} restricts the length of waveguides. Constraint~\eqref{eq:multicast-position_constraint2} guarantees the minimum spacing among PAs to be no smaller than $\Delta$ for avoiding antenna coupling~\cite{10981775}. Constraint~\eqref{eq:multicast-time-allocation-constraint} is the time allocation constraint, which is only valid when the WS structure is applied. To transform problem~\eqref{eq:multicast-optimization-problem} into a more tractable form, we rewrite~\eqref{eq:multicast-optimization-problem} as follows:
\begin{subequations}
\label{eq:multicast-optimization-problem-2}
\begin{equation}
\label{eq:multicast-objective-function-2}
\mathcal{P}_{\text{MC}}(1): \max_{\left\{\mathbf{X},\mathbf{W},\mathbf{p},\boldsymbol{\lambda},R_{\text{min}}\right\}}R_{\text{min}},
\end{equation}
\begin{equation}
    {\rm{s.t.}} \ \ R_{\text{MC},k,g}^{\text{Y}}\geq R_{\text{min}}, \forall k\in\mathcal{K}_{\text{MC}}, \forall g\in\mathcal{G}_{\text{k}},
\end{equation}
\begin{equation}
\eqref{eq:multicast-maximum-power}-\eqref{eq:multicast-time-allocation-constraint}.
\end{equation}
\end{subequations}

\subsection{Discussions}
\textcolor{blue}{For unicast communications, consider there exist $K$ users, with $\mathcal{K}_{\text{UC}}=\left\{1, \dots, K\right\}$ representing the set of users. Let $\boldsymbol{\psi}_{\text{u}}^{k}=\left[x_{\text{u}}^{k},y_{\text{u}}^{k},0\right]$ denote the location of the $k$-th user.} Then, the wireless channel between PAs on the $i$-th waveguide and the user $k\in\mathcal{K}_{\text{UC}}$, i.e., $\mathbf{h}_{k}^{i}\left(\mathbf{x}_{\text{p}}^{i}\right)$, can be expressed in a similar way as in~\eqref{eq:h-channel} by substituting $\boldsymbol{\psi}_{\text{u}}^{k,g}$ with $\boldsymbol{\psi}_{\text{u}}^{k}$. Correspondingly, the channel vector from all PAs and the $k$-th user is represented by $\mathbf{h}_{k}\left(\mathbf{X}\right)=\left[\mathbf{h}^{1}_{k}\left(\mathbf{x}_{\text{p}}^1\right), \dots,  \mathbf{h}^{K}_{k}\left(\mathbf{x}_{\text{p}}^K\right)\right]\in\mathbb{C}^{1\times M}$. For WM, WD, and WS, the achievable rate of user $k$, denoted by $R_{\text{UC},k}^{\text{WM}}$, $R_{\text{UC},k}^{\text{WD}}$, and $R_{\text{UC},k}^{\text{WS}}$, can be given by substituting $\mathbf{h}_{k,g}\left(\mathbf{X}\right)$, $\mathbf{h}_{k,g}^k$, and $\mathbf{h}_{k,g}^{k'}$ with $\mathbf{h}_{k}\left(\mathbf{X}\right)$, $\mathbf{h}_{k}^k$, and $\mathbf{h}_{k}^{k'}$, in~\eqref{eq:rate-WM},~\eqref{eq:rate-WD}, and~\eqref{eq:rate-TS}, respectively. Hence, the MMF problem for the unicast case can be formulated as follows:
\begin{subequations}
\label{eq:unicast-optimization-problem}
\begin{equation}
\label{eq:unicast-objective-function}
\mathcal{P}_{\text{UC}}(0): \max_{\left\{\mathbf{X},\mathbf{W},\mathbf{p},\boldsymbol{\lambda}\right\}}\min_{k\in\mathcal{K}_{\text{UC}}}R_{\text{UC},k}^{\text{Y}},
\end{equation}
\begin{equation}
{\rm{s.t.}} \ \ \eqref{eq:multicast-maximum-power}-\eqref{eq:multicast-time-allocation-constraint}.
\end{equation}
\end{subequations}
Again, to improve the tractability of problem $\mathcal{P}_{\text{UC}}(0)$, we rewrite it into the following equivalent one:
\begin{subequations}
\label{eq:unicast-optimization-problem-2}
\begin{equation}
\label{eq:unicast-objective-function-2}
\mathcal{P}_{\text{UC}}(1): \max_{\left\{\mathbf{X},\mathbf{W},\mathbf{p},\boldsymbol{\lambda},R_{\text{min}}\right\}}R_{\text{min}},
\end{equation}
\begin{equation}
\label{eq:minimum-rate-constraint-unicast}
    {\rm{s.t.}} \ \ R_{\text{UC},k}^{\text{Y}}\geq R_{\text{min}}, \forall k\in\mathcal{K}_{\text{UC}},
\end{equation}
\begin{equation}
\eqref{eq:multicast-maximum-power}-\eqref{eq:multicast-time-allocation-constraint}.
\end{equation}
\end{subequations}

Notably, the formulated optimization problem~\eqref{eq:unicast-optimization-problem-2} for unicast communications can be regarded as a special case of the optimization problem~\eqref{eq:multicast-optimization-problem-2} for multicast communications. More specifically, problem~\eqref{eq:unicast-optimization-problem-2} can be obtained from problem~\eqref{eq:multicast-optimization-problem-2} by setting the number of users in each multicast group to one, i.e., $\left|\mathcal{G}_k\right|=1,\forall k$. Motivated by this, we focus on the optimization problem~\eqref{eq:multicast-optimization-problem-2} for multicast communications, while problem~\eqref{eq:unicast-optimization-problem-2} for unicast communications can be solved in a similar manner. However, problem~\eqref{eq:multicast-optimization-problem-2} is intractable to solve since the constraint~\eqref{eq:minimum-rate-constraint-unicast} is highly non-convex w.r.t. $\mathbf{X}$ due to the complicated exponential components. Moreover, the involved variables $\mathbf{X}$ and $\left\{\mathbf{w}_k\right\}$ are highly coupled for the WM and TS structures. Generally, there is no established approach for efficiently solving such non-convex optimization problem. In the following, we invoke the PDD method for alternatively updating multi-block variables, which is guaranteed to converge to the final stable solution efficiently.

\section{Proposed Solutions}
In this section, we first propose a PDD-based iterative algorithm to solve the MMF problem for WM, which is further extended to solve the problems for WD and WS. 
\subsection{Proposed Solution for WM}
To deal with the fractional format of the SINR expression in problem~\eqref{eq:multicast-optimization-problem-2} for WM, we first equivalently transform it into a more tractable form with \textbf{Lemma~\ref{lemma:SINR-transform}}.
\begin{lemma}
\label{lemma:SINR-transform}
The problem~\eqref{eq:multicast-optimization-problem-2} for WM is equivalent to the following one:
\begin{subequations}
\label{eq:multicast-optimization-problem-WM2}
\begin{equation}
\label{eq:multicast-objective-function-WM2}
\max_{\left\{\mathbf{X},\mathbf{W},\gamma\right\}}\gamma,
\end{equation}
\begin{equation}
\label{eq:multicast_rate_constraint-WM2}
{\rm{s.t.}} \ \ y\left(\mu_{k,g},\mathbf{X},\mathbf{W}\right)\geq \gamma, \forall k\in\mathcal{K}_{\text{MC}}, \forall g\in\mathcal{G}_{\text{k}},
\end{equation}
\begin{equation}
\label{eq:multicast_power_constraint-WM1}
\sum_{k=1}^K\left\|\mathbf{w}_k\right\|^2\leq P_{\text{max}},
\end{equation}
\begin{equation}
\eqref{eq:multicast-position_constraint1}, \eqref{eq:multicast-position_constraint2}.
\end{equation}
\end{subequations}
In~\eqref{eq:multicast-optimization-problem-WM2}, $y(\mu_{k,g},\mathbf{X},\mathbf{W})$ is given by
\begin{align}
y&(\mu_{k,g},\mathbf{X},\mathbf{W})=2\Re\left\{\mu_{k,g}^*\mathbf{h}_{k,g}\left(\mathbf{X}\right)\mathbf{G}\left(\mathbf{X}\right)\mathbf{w}_k\right\}\nonumber\\&-|\mu_{k,g}|^2\left(\sum\limits_{k'\neq k}^K  \left|{\mathbf{h}_{k,g}\left(\mathbf{X}\right)\mathbf{G}\left(\mathbf{X}\right)}\mathbf{w}_{k'}\right|^2+\sigma_k^2\right),
\label{eq:y-mu}
\end{align}
where $\boldsymbol{\mu}=\{\mu_{k,g}\}_{\forall k,g}$ with $\mu_{k,g}$ given by
\begin{align}
\label{eq:optimal_mu}
\mu_{k,g}=\frac{\mathbf{h}_{k,g}\left(\mathbf{X}\right)\mathbf{G}\left(\mathbf{X}\right)\mathbf{w}_k}{\sum\limits_{k'\neq k}^K  \left|{\mathbf{h}_{k,g}\left(\mathbf{X}\right)\mathbf{G}\left(\mathbf{X}\right)}\mathbf{w}_{k'}\right|^2+\sigma_k^2}.
\end{align}
\end{lemma}
\begin{proof}
    See Appendix~\ref{app:lemma-1-proof}.
\end{proof}

Observing~\eqref{eq:y-mu}, we can find that $y\left(\mu_{k,g},\mathbf{X},\mathbf{W}\right)$ is highly complex w.r.t. $\mathbf{X}$ due to the exponential and fractional expressions, which makes it hard to directly apply existing optimization methods for solving~\eqref{eq:multicast-optimization-problem-WM2}. To tackle this issue, we first introduce auxiliary variables  $e_{k,g}^{i,n}=\frac{2\pi}{\lambda}\left\|\boldsymbol{\psi}^{k,g}_{\text{u}}-\boldsymbol{\psi}_{\text{p}}^{i,n}\right\| + \frac{2\pi}{\lambda_g}{x}_{\text{p}}^{i,n}$ and $u_{k,g}^{i,n}=\frac{{\eta} e^{-je_{k,g}^{i,n}}}{\sqrt{N}\left\|\boldsymbol{\psi}^{k,g}_{\text{u}}-\boldsymbol{\psi}_{\text{p}}^{i,n}\right\|}$, $\forall i\in\mathcal{K}_{\text{WG}}, n\in\mathcal{N}_i, k\in\mathcal{K}_{\text{MC}}, g\in\mathcal{G}_k$. Then, we have the following equality constraints:
\begin{subequations}
\begin{equation}
a_{k,g}^{i,n}\triangleq u_{k,g}^{i,n}\left\|\boldsymbol{\psi}^{k,g}_{\text{u}}-\boldsymbol{\psi}_{\text{p}}^{i,n}\right\| -{\frac{\eta}{ \sqrt{N}}} e^{-je_{k,g}^{i,n}}=0, \forall k, g, i, n,
\label{eq:u-equality}
\end{equation}
\begin{equation}
\label{eq:e-equality}
b_{k,g}^{i,n}\triangleq e_{k,g}^{i,n}-\left(\frac{2\pi}{\lambda}\left\|\boldsymbol{\psi}^{k,g}_{\text{u}}-\boldsymbol{\psi}_{\text{p}}^{i,n}\right\| + \frac{2\pi}{\lambda_g}{x}_{\text{p}}^{i,n}\right)=0, \forall k, g, i, n.
\end{equation}
\label{eq:equality_constraint}
\end{subequations} 
Define variables $\mathbf{e}_{k,g}\in\mathbb{R}^{K\times N}$ with the $(i,n)$-th element being $e_{k,g}^{i,n}$, and $\left(\mathbf{u}_{k,g}^{i}\right)^T=\begin{bmatrix}
u_{k,g}^{i,1},u_{k,g}^{i,2},\ldots,u_{k,g}^{i,N}
\end{bmatrix}\in\mathbb{C}^{1\times N}$ is the pinching beamforming vector of user $g$ in group $k$ using PAs in $i$-th waveguide. $\mathbf{u}_{k,g}^T=\begin{bmatrix}
\left(\mathbf{u}_{k,g}^{1}\right)^T,\left(\mathbf{u}_{k,g}^{2}\right)^T,\ldots,\left(\mathbf{u}_{k,g}^{K}\right)^T
\end{bmatrix}\in\mathbb{C}^{1\times M}$ stacks the pinching beamforming coefficients for all the PAs. Furthermore, denote {$\mathbf{E}=\left\{\mathbf{e}_{k,g}\right\}_{\forall k,g}$,} $\mathbf{U}=\left\{\mathbf{u}_{k,g}\right\}_{\forall k,g}$. Then, problem~\eqref{eq:multicast-optimization-problem-WM2} can be equivalently written as
\begin{subequations}
\label{eq:multicast-optimization-problem-WM3}
\begin{equation}
\label{eq:multicast-objective-function-WM3}
\max_{\left\{\mathbf{X},\mathbf{W},\boldsymbol{\mu},\gamma,\mathbf{U},\mathbf{E}\right\}}\gamma,
\end{equation}
\begin{align}
\label{eq:multicast_rate_constraint-WM3}
{\rm{s.t.}} \ \ y(\mu_{k,g},\mathbf{u}_{k,g},\mathbf{w}_k)\geq \gamma,\forall k,g,
\end{align}
\begin{equation}
\eqref{eq:multicast-position_constraint1}, \eqref{eq:multicast-position_constraint2}, \eqref{eq:multicast_power_constraint-WM1}, \eqref{eq:u-equality}, \eqref{eq:e-equality},
\end{equation}
\end{subequations}
where $y(\mu_{k,g},\mathbf{u}_{k,g},\mathbf{w}_k)$ in~\eqref{eq:multicast_rate_constraint-WM2} is given by 
\begin{align}
y(\mu_{k,g},\mathbf{u}_{k,g},&\mathbf{w}_k)=2\Re\left\{\mu_{k,g}^*\mathbf{u}_{k,g}\boldsymbol{\Sigma}\mathbf{w}_k\right\}\nonumber\\
&-|\mu_{k,g}|^2\left(\sum\limits_{k'\neq k}^K  \left|\mathbf{u}_{k,g}\boldsymbol{\Sigma}\mathbf{w}_{k'}\right|^2+\sigma_k^2\right),
\end{align}
and $\boldsymbol{\Sigma}=\mathrm{blkdiag}(\mathbf{1}_{N\times1},\mathbf{1}_{N\times1},\ldots,\mathbf{1}_{N\times1})\in\mathbb{R}^{M\times K}$ denotes a block diagonal matrix with all non-zero elements being $1$. The conventional AO method can be applied for solving problem~\eqref{eq:multicast-optimization-problem-WM3} by decomposing it into a sequence of subproblems, which however can not guarantee the convergence to stationary solutions due to the equality coupling constraints~\eqref{eq:u-equality} and \eqref{eq:e-equality}~\cite{9120361,9119203}. To tackle this specific challenge, in the following, we employ the PDD algorithm that alleviates the stringent coupling constraints through augmented Lagrangian relaxation.

Denote by $\mathcal{L}\left(\mathcal{V};\left\{\boldsymbol{\lambda}_{k,g}^u\right\}, \left\{\boldsymbol{\lambda}_{k,g}^e\right\}, \rho\right)$ the augmented Lagrangian function with penalty factor $\rho$ and dual variables $\boldsymbol{\lambda}_{k,g}^u\in\mathbb{C}^{K\times N}, \boldsymbol{\lambda}_{k,g}^e\in\mathbb{C}^{K\times N}$ corresponding to the equality constraints~\eqref{eq:u-equality} and \eqref{eq:e-equality}, respectively, where $\mathcal{V}=\left\{\mathbf{X},\mathbf{U},\mathbf{E}\right\}$. With the definition of $\mathcal{L}\left(\mathcal{V};\left\{\boldsymbol{\lambda}_{k,g}^u\right\}, \left\{\boldsymbol{\lambda}_{k,g}^e\right\}, \rho\right)$ given in~\eqref{eq:Lagrangian_function} at the top of the next page, where $\mathbf{A}_{k,g}^{u}=\left\{a_{k,g}^{i,n}\right\}$ and $\mathbf{B}_{k,g}^{e}=\left\{b_{k,g}^{i,n}\right\}$ stack the residuals of the equality constraints in~\eqref{eq:equality_constraint},
\begin{figure*}
\begin{align}
\label{eq:Lagrangian_function}
\mathcal{L}\left(\mathcal{V};\left\{\boldsymbol{\lambda}_{k,g}^u\right\}, \left\{\boldsymbol{\lambda}_{k,g}^e\right\}, \rho\right)=\sum_{k=1}^K\sum_{g=1}^G\left(\left\|\mathbf{A}_{k,g}^{u}+\rho\boldsymbol{\lambda}_{k,g}^u\right\|^2+\left\|\mathbf{B}_{k,g}^{e}+\rho\boldsymbol{\lambda}_{k,g}^e\right\|^2\right).
\end{align}
\hrulefill
\end{figure*}
problem~\eqref{eq:multicast-optimization-problem-WM3} can be equivalently transformed to the following AL problem: 
\begin{subequations}
\label{eq:multicast-optimization-problem-AL}
\begin{align}
\label{eq:multicast-objective-function-AL}
\max_{\left\{\mathbf{X},\mathbf{W},\boldsymbol{\mu},\gamma,\mathbf{U},\mathbf{E}\right\}}\gamma&-\frac{1}{2\rho}\times \mathcal{L}\left(\mathcal{V};\left\{\boldsymbol{\lambda}_{k,g}^u\right\}, \left\{\boldsymbol{\lambda}_{k,g}^e\right\}, \rho\right),
\end{align}
\begin{equation}
\label{eq:multicast_rate_constraint-AL}
{\rm{s.t.}} \ \ y(\mu_{k,g},\mathbf{u}_{k,g},\mathbf{w}_k)\geq \gamma,\forall k,g,
\end{equation}
\begin{equation}
\eqref{eq:multicast-position_constraint1}, \eqref{eq:multicast-position_constraint2},\eqref{eq:multicast_power_constraint-WM1}.
\end{equation}
\end{subequations}
For solving problem~\eqref{eq:multicast-optimization-problem-AL}, we optimize the following subproblems in an iterative manner.

\subsubsection{Subproblem w.r.t. $\left\{\gamma,\mathbf{W}\right\}$}
The auxiliary variable $\gamma$ and the baseband beamforming $\mathbf{W}$ are jointly optimized by solving the following constrained optimization problem:
\begin{subequations}
\label{eq:multicast-optimization-subproblem-WQ}
\begin{align}
\label{eq:multicast-objective-function-WQ}
\max_{\left\{\gamma,\mathbf{W}\right\}}\gamma,
\end{align}
\begin{equation}
\label{eq:multicast_rate_constraint-WQ}
{\rm{s.t.}} \ \ y\left(\mu_{k,g},\mathbf{W}\right)\geq \gamma,\forall k,g,
\end{equation}
\begin{equation}
\eqref{eq:multicast_power_constraint-WM1}.
\end{equation}
\end{subequations}
The above problem is convex, which can be efficiently solved using standard convex problem solvers such as CVX~\cite{boyd2004convex}.

\subsubsection{Subproblem w.r.t. $\left\{\gamma,\mathbf{U}\right\}$}
The auxiliary variable $\gamma$ and $\mathbf{U}$ can be optimized by solving the following problem:
\begin{subequations}
\label{eq:multicast-optimization-subproblem-U}
\begin{align}
\label{eq:multicast-objective-function-U}
\max_{\left\{\gamma, \mathbf{U}\right\}}\gamma-\sum_{k=1}^K\sum_{g=1}^G\left\|\mathbf{R}_{k,g}\mathbf{u}_{k,g}+\boldsymbol{\zeta}_{k,g}\right\|^2,
\end{align}
\begin{equation}
\label{eq:multicast_rate_constraint-U}
{\rm{s.t.}} \ \ y\left(\mu_{k,g},\mathbf{U}\right)\geq \gamma,\forall k,g,
\end{equation}
\end{subequations}
where $\mathbf{R}_{k,g}=\mathrm{diag}(r_{k,g}^{1,1},r_{k,g}^{1,2},\ldots,r_{k,g}^{k,g})\in\mathbb{R}^{M\times M}$ with $r_{k,g}^{i,n}=\left\|\boldsymbol{\psi}^{k,g}_{\text{u}}-\boldsymbol{\psi}_{\text{p}}^{i,n}\right\|$. Moreover, $\zeta_{k,g}^{i,n}=\rho\lambda_{k,g,i,n}^{u}-\sqrt{\frac{\eta}{N}} e^{-je_{k,g}^{i,n}}$, and vector $\zeta_{k,g}=\left[\zeta_{k,g}^{1,1},\zeta_{k,g}^{1,2},\ldots,\zeta_{k,g}^{K,N}\right]^{T}\in\mathbb{C}^{M\times1}$. \eqref{eq:multicast-optimization-subproblem-U} is also a convex optimization problem, which can be solved using CVX.

\subsubsection{Subproblem w.r.t. $\mathbf{X}$} Since the optimization variable $\mathbf{X}$ appears only in the augmented
Lagrangian term in~\eqref{eq:multicast-objective-function-AL}, the subproblem for optimizing $\mathbf{X}$ is given by
\begin{subequations}
\label{eq:multicast-optimization-subproblem-AL-X}
\begin{align}
\label{eq:multicast-objective-function-AL-X}
\min_{\mathbf{X}}\mathcal{L}\left(\mathcal{V};\left\{\boldsymbol{\lambda}_{k,g}^u\right\}, \left\{\boldsymbol{\lambda}_{k,g}^e\right\}, \rho\right),
\end{align}
\begin{equation}
\eqref{eq:multicast-position_constraint1},\eqref{eq:multicast-position_constraint2}.
\end{equation}
\end{subequations}
For solving the non-convex problem~\eqref{eq:multicast-optimization-subproblem-AL-X}, we invoke the successive convex approximation (SCA) method for iteratively updating $\mathbf{X}$. Specifically, we first rewrite $\mathcal{L}\left(\mathcal{V};\left\{\boldsymbol{\lambda}_{k,g}^u\right\}, \left\{\boldsymbol{\lambda}_{k,g}^e\right\}, \rho\right)$ as a function of ${x}_{\text{p}}^{i,n}$, i.e.,
\begin{align}
\label{eq:AL_function_x}
&\mathcal{L}\left(\mathcal{V};\left\{\boldsymbol{\lambda}_{k,g}^u\right\}, \left\{\boldsymbol{\lambda}_{k,g}^e\right\}, \rho\right)\nonumber\\&=\sum_{k=1}^K\sum_{g=1}^G\sum_{i=1}^K\sum_{n=1}^N\left(L_{k,g,i,n}^{\text{DC}}\left({x}_{\text{p}}^{i,n}\right)+L_{k,g,i,n}^{\text{NC}}\left({x}_{\text{p}}^{i,n}\right)\right),
\end{align}
where $L_{k,g,i,n}^{\text{DC}}\left({x}_{\text{p}}^{i,n}\right)$ and $L_{k,g,i,n}^{\text{NC}}\left({x}_{\text{p}}^{i,n}\right)$ are given in~\eqref{eq:DC_function} and~\eqref{eq:NC_function}, respectively, with $L_{k,g,i,n}^{\text{CC}}\left(x_{\text{p}}^{i,n}\right)=\Omega_{k,g}^{i,n}r_{k,g}^{i,n}$
and 
\begin{align}
\Omega_{k,g}^{i,n}\triangleq \Re\left\{\left(u_{k,g}^{i,n}\right)^{H}\left(\lambda_{k,g,i,n}^{u}-\frac{\sqrt{\eta} e^{-je_{k,g}^{i,n}}}{\sqrt{N}\rho}\right) \right.\nonumber\\ \left. -\frac{2\pi}{\lambda}\left(\lambda_{k,g,i,n}^{e}+\frac{\lambda}{2\pi}\frac{e_{k,g}^{i,n}}{\rho}\right)\right\} \nonumber.
\end{align}
\begin{figure*}
\begin{equation}
\label{eq:DC_function}
L_{k,g,i,n}^{\text{DC}}\left(x_{\text{p}}^{i,n}\right)=\frac{1}{2\rho}\left(\left(u_{k,g}^{i,n}\right)^{2}+1\right)\left(x_{\text{u}}^{k,g}-x_{\text{p}}^{i,n}\right)^{2}+\frac{1}{2\rho}\left(\frac{2\pi}{\lambda_g}{x}_{\text{p}}^{i,n}\right)^{2}-\frac{2\pi}{\lambda}\left(\lambda_{k,g,i,n}^{e}+\frac{e_{k,g}^{i,n}}{\rho}\right)+L_{k,g,i,n}^{\text{CC}}\left(x_{\text{p}}^{i,n}\right).
\end{equation}
\end{figure*}
\begin{figure*}
\begin{equation}
\label{eq:NC_function}
L_{k,g,i,n}^{\text{NC}}\left(x_{\text{p}}^{i,n}\right)=\left(\frac{2\pi}{\lambda}\right)^2\frac{n_{\text{eff}}}{\rho}x_{\text{p}}^{i,n}\sqrt{\left(x_{\text{u}}^{k,g}-x_{\text{p}}^{i,n}\right)^2+\left(y_{\text{u}}^{k,g}-y_{\text{p}}^{i}\right)^2+d^2}.
\end{equation}
\end{figure*}
Since $L_{k,g,i,n}^{\text{DC}}\left({x}_{\text{p}}^{i,n}\right)$ is a difference of convex (D.C.) functions, in the $t$-th SCA iteration, we apply the first-order Taylor expansion over $L_{k,g,i,n}^{\text{CC}}\left(x_{\text{p}}^{i,n}\right)$ to obtain its tight upper bound $\widehat{L}^{\mathrm{CC}}_{k,g,i,n}$ as presented in~\eqref{eq:DC_upper},
\begin{figure*}
\begin{align}
\label{eq:DC_upper}
\widehat{L}^{\mathrm{CC}}_{k,g,i,n}\left(x_{\text{p}}^{i,n}\right)=L_{k,g,i,n}^{\text{CC}}\left(\left(x_{\text{p}}^{i,n}\right)^{\left(t\right)}\right)+\nabla_{x_{\text{p}}^{i,n}}L^{\mathrm{CC}}_{k,g,i,n}\left(\left(x_{\text{p}}^{i,n}\right)^{\left(t\right)}\right)\left(x_{\text{p}}^{i,n}-\left(x_{\text{p}}^{i,n}\right)^{\left(t\right)}\right).
\end{align}
\end{figure*}
where $\nabla_{x_{\text{p}}^{i,n}}L^{\mathrm{CC}}_{k,g,i,n}$ is the derivative of $L^{\mathrm{CC}}_{k,g,i,n}$ w.r.t. $x_{\text{p}}^{i,n}$. 
Then, we can obtain the convex upper bound of {$L_{k,g,i,n}^{\text{DC}}\left({x}_{\text{p}}^{i,n}\right)$, denoted by $\widehat{L}_{k,g,i,n}^{\text{DC}}\left({x}_{\text{p}}^{i,n}\right)$}, by substituting $L_{k,g,i,n}^{\text{CC}}\left(x_{\text{p}}^{i,n}\right)$ with $\widehat{L}_{k,g,i,n}^{\text{CC}}\left(x_{\text{p}}^{i,n}\right)$ in~\eqref{eq:DC_function}. 

For the non-convex component $L_{k,g,i,n}^{\text{NC}}\left(x_{\text{p}}^{i,n}\right)$, we utilize the Jensen’s inequality to obtain the convex surrogate function. Specifically, {the Jensen’s inequality is given by
\begin{align}
\label{eq:Jensen_inequality}
\sqrt{y}\leq\sqrt{y^{(t)}}+\frac{y-y^{(t)}}{2\sqrt{y^{(t)}}}\leq\frac{y+y^{(t)}}{2\sqrt{y^{(t)}}}. 
\end{align}
Substituting $y=x_{\text{p}}^{i,n}\sqrt{\left(x_{\text{u}}^{k,g}-x_{\text{p}}^{i,n}\right)^2+\left(y_{\text{u}}^{k,g}-y_{\text{p}}^{i}\right)^2+d^2}$ into~\eqref{eq:Jensen_inequality}, the upper bound of $L_{k,g,i,n}^{\text{NC}}\left(x_{\text{p}}^{i,n}\right)$ is given in~\eqref{eq:NC_upper},
\begin{figure*}
    \begin{equation}
\label{eq:NC_upper}
\widehat{L}_{k,g,i,n}^{\text{NC}}\left(x_{\text{p}}^{i,n}\right)=\left(\frac{2\pi}{\lambda}\right)^2\frac{n_{\text{eff}}}{\rho}\frac{\left(x_{\text{p}}^{i,n}\right)^3+q_{k,g}^{i,n}x_{\text{p}}^{i,n}+2x_{\text{u}}^{k,g}\left(\left(x_{\text{p}}^{i,n}\right)^{\left(t\right)}\right)^2}{2\sqrt{\left(\left(x_{\text{p}}^{i,n}\right)^{\left(t\right)}-x_{\text{u}}^{k,g}\right)^2+\left(y_{\text{p}}^{i}-y_{\text{u}}^{k,g}\right)^2+d^2}}.
\end{equation}
\hrulefill
\end{figure*}
shown at the top of the next page, where $q_{k,g}^{i,n}$ is fixed at each iteration, defined as
\begin{align}
q_{k,g}^{i,n}&=\left(x_{\text{u}}^{k,g}\right)^{2}+2\left(y_{\text{u}}^{k,g}-y_{\text{p}}^{i}\right)^2+2d^2\nonumber\\&+\left(\left(x_{\text{p}}^{i,n}\right)^{\left(t\right)}-x_{\text{u}}^{k,g}\right)^{2}-4x_{\text{u}}^{k,g}\left(x_{\text{p}}^{i,n}\right)^{\left(t\right)}
\end{align}
} 

With the above approximations, the optimization problem~\eqref{eq:multicast-optimization-subproblem-AL-X} can be reformulated as the following convex one in the $t$-th SCA iteration:
\begin{subequations}
\label{eq:multicast-optimization-subproblem-X}
\begin{align}
\label{eq:multicast-objective-function-X}
\min_{\mathbf{X}}\sum_{k=1}^K\sum_{g=1}^G\sum_{i=1}^K\sum_{n=1}^N\left(\widehat{L}^{\mathrm{DC}}_{k,g,i,n}\left(x_{\text{p}}^{i,n}\right)+\widehat{L}^{\mathrm{NC}}_{k,g,i,n}\left(x_{\text{p}}^{i,n}\right)\right),
\end{align}
\begin{equation}
\eqref{eq:multicast-position_constraint1},\eqref{eq:multicast-position_constraint2},
\end{equation}
\end{subequations}
which can be solved with the CVX\textcolor{blue}{\footnote{\textcolor{blue}{Due to the multimodal property of problem~\eqref{eq:multicast-optimization-subproblem-AL-X}, we perform multiple initializations of PAs positions and select the one yielding the best performance.}}}.

\begin{algorithm}[h]
\renewcommand{\algorithmicrequire}{\textbf{Input:}}
\renewcommand{\algorithmicensure}{\textbf{Output:}}
\caption{PDD Algorithm for Solving Problem~\eqref{eq:multicast-optimization-problem-WM2}}
\label{alg:PDD-WM}
    \begin{algorithmic}[1]
    \State Initialize optimization variables $\mathbf{X}$, $\mathbf{W}$, $\mathbf{U}$, $\mathbf{E}$, $\mu$, and $\gamma$. 
    \State Initialize penalty parameter $\rho^{\left(0\right)}$, dual variables $\left\{\boldsymbol{\lambda}_{k,g}^u\right\}^{\left(0\right)}, \left\{\boldsymbol{\lambda}_{k,g}^e\right\}^{\left(0\right)}$.
    \State Initialize $l=0$.
    \While{the maximum constraint residual, i.e., $\max\left\{\left\|\mathbf{A}^{(l)}\right\|_\infty,\left\|\mathbf{B}^{(l)}\right\|_\infty\right\}$, is greater than the threshold $\epsilon_1$}
    \While{the fractional increment of the objective function is greater than the threshold $\epsilon_2$}
    \State Update $\mu_{k,g}$ with~\eqref{eq:optimal_mu}, $\forall k,g$.
    \State Update $\left\{\gamma,\mathbf{W}\right\}$ by solving~\eqref{eq:multicast-optimization-subproblem-WQ}.
    \State Update $\left\{\gamma,\mathbf{U}\right\}$ by solving~\eqref{eq:multicast-optimization-subproblem-U}.
    \State Update $\textbf{X}$ by solving~\eqref{eq:multicast-optimization-subproblem-X}.
    \State Update $\mathbf{E}$ with~\eqref{eq:close-theta}.
    \EndWhile
\If{$\left\|\mathbf{A}^{(l)}\left(\mathbf{B}^{(l)}\right)\right\|_{\infty}\leqslant0.9\left\|\mathbf{A}^{(l)}\left(\mathbf{B}^{(l)}\right)\right\|_{\infty}$}
    \State Update $\boldsymbol{\lambda}^{(l+1)}=\ \boldsymbol{\lambda}^{(l)}+\frac{1}{\rho^{(l)}}\mathbf{A}^{(l)}\left(\mathbf{B}^{(l)}\right)$
    \Else 
    \State Update $\rho^{(l+1)}=\rho^{(l)}\times0.85$.
    \EndIf
    \EndWhile
    \Ensure $\mathbf{X},\mathbf{W}$.
    \end{algorithmic}
\end{algorithm}

\subsubsection{Subproblem w.r.t. $\mathbf{E}$}
The subproblem for optimizing $\mathbf{E}$ is given by

{
\begin{align}
\label{eq:multicast-objective-function-AL-E}
\min_{\mathbf{E}}\sum_{k=1}^K\sum_{g=1}^G\sum_{i=1}^K\sum_{n=1}^NL_{k,g,i,n}^{\text{CV}}\left(e_{k,g}^{i,n}\right)+L_{k,g,i,n}^{\text{EXP}}\left(e_{k,g}^{i,n}\right),
\end{align}
}
{where $L_{k,g,i,n}^{\text{CV}}\left(e_{k,g}^{i,n}\right)$ and $L_{k,g,i,n}^{\text{EXP}}\left(e_{k,g}^{i,n}\right)$ are given by
\begin{align}
\label{eq:function_CV}
L_{k,g,i,n}^{\text{CV}}\left(e_{k,g}^{i,n}\right)=\frac{1}{2\rho}\left(e_{k,g}^{i,n}-\left(\frac{2\pi}{\lambda}r_{k,g}^{i,n}  \right. \right.\nonumber\\
\left. \left.+  \frac{2\pi}{\lambda_g}{x}_{\text{p}}^{i,n}\right)+\rho\lambda_{k,g,i,n}^e\right)^2,
\end{align}
and
\begin{align}
\label{eq:function_EXP}
L_{k,g,i,n}^{\text{EXP}}\left(e_{k,g}^{i,n}\right)=- \sqrt{\frac{\eta}{N}} \Re\left\{\left(\lambda_{k,g,i,n}^{u}+\frac{u_{k,g}^{i,n}}{\rho}r_{k,g}^{i,n}\right)e^{je_{k,g}^{i,n}}\right\},
\end{align}
respectively. To address the non-convexity of $L_{k,g,i,n}^{\text{EXP}}\left(e_{k,g}^{i,n}\right)$, we utilize the SCA method again for iteratively updating $\mathbf{E}$.} Specifically, by utilizing the Lipschitz gradient surrogate, the convex upper bound of $L_{k,g,i,n}^{\text{EXP}}$ is given in~\eqref{eq:EXP_upper},
\begin{figure*}[t]
\begin{equation}
\label{eq:EXP_upper}
\widehat{L}_{k,g,i,n}^{\mathrm{EXP}}\left(e_{k,g}^{i,n}\right)=L_{k,g,i,n}^{\mathrm{EXP}}\left(\left(e_{k,g}^{i,n}\right)^{(t)}\right)+\nabla_{e_{k,g}^{i,n}}L_{k,g,i,n}^{\mathrm{EXP}}\left(\left(e_{k,g}^{i,n}\right)^{(t)}\right)\left(e_{k,g}^{i,n}-\left(e_{k,g}^{i,n}\right)^{(t)}\right)+\left|\frac{\varrho_{k,g}^{i,n}}{2}\left(e_{k,g}^{i,n}-\left(e_{k,g}^{i,n}\right)^{(t)}\right)\right|^{2}.
\end{equation}
\end{figure*}
shown at the top of the next page, where the corresponding Lipschitz constant can be given by
\begin{align}
\varrho_{k,g,i,n}^e=\sqrt{\frac{\eta}{N}}\left|\lambda_{k,g,i,n}^u+\frac{u_{k,g}^{i,n}}{\rho}r_{k,g}^{i,n}\right|.
\end{align}
The derivation for $\varrho_{k,g,i,n}^e$ is omitted here for the sake of brevity. The reader is referred to~\cite{xu2013block} for more details. By substituting $L_{k,g,i,n}^{\text{EXP}}\left(e_{k,g}^{i,n}\right)$ with $\widehat{L}_{k,g,i,n}^{\mathrm{EXP}}\left(e_{k,g}^{i,n}\right)$ in~\eqref{eq:multicast-objective-function-AL-E}, we can obtain the closed-form solution of $e_{k,g}^{i,n}$ in the $t$-th SCA iteration as given in~\eqref{eq:close-theta}, shown at the top of the next page.
\begin{figure*}
\begin{align}
\label{eq:close-theta}
e_{k,g}^{i,n}=\frac{\varrho_{k,g,i,n}^{e}\left(e_{k,g}^{i,n}\right)^{(t)}-\lambda_{k,g,i,n}^{e}+\frac{1}{\rho}\left(\frac{2\pi}{\lambda}r_{k,g}^{i,n} + \frac{2\pi}{\lambda_g}{x}_{\text{p}}^{i,n}\right)-\nabla_{e}L^{\mathrm{EXP}}\left(\left(e_{k,g}^{i,n}\right)^{(t)}\right)}{\varrho_{k,g,i,n}^{e}+\frac{1}{\rho}}.
\end{align}
\hrulefill
\end{figure*}

\subsubsection{Overall PDD algorithm and property analysis}
The overall PDD algorithm for solving~\eqref{eq:multicast-optimization-problem-WM2} consists of a nested loop, as shown in \textbf{Algorithm~\ref{alg:PDD-WM}}. The penalty factor $\rho$ and dual variables $\boldsymbol{\lambda}_{k,g}^u, \boldsymbol{\lambda}_{k,g}^e$ are updated in the outer loop, while the four subproblems are iteratively updated in the inner loop. The convergence of \textbf{Algorithm~\ref{alg:PDD-WM}} is analyzed as follows. For the inner loop, since the minimum data rate is upper bounded and guaranteed to increase after solving each subproblem, the inner loop can converge within limited iterations. The outer loop is guaranteed to converge, as the AL function $\mathcal{L}\left(\mathcal{V};\left\{\boldsymbol{\lambda}_{k,g}^u\right\}, \left\{\boldsymbol{\lambda}_{k,g}^e\right\}, \rho\right)$ is non-increasing in the procedure for solving the subproblems.

We now present the complexity analysis of \textbf{Algorithm~\ref{alg:PDD-WM}}. The complexity for updating $\mu_{k,g}$ is given by $\mathcal{O}(K^2G)$. Moreover, if
the interior point method is employed, the computational complexity for solving problem~\eqref{eq:multicast-optimization-subproblem-WQ},~\eqref{eq:multicast-optimization-subproblem-U} and~\eqref{eq:multicast-optimization-subproblem-X} are $\mathcal{O}(K^3)$, $\mathcal{O}\left(\left(KGM\right)^{1.5}\right)$ and $\mathcal{O}\left(\left(KGM\right)^{1.5}\right)$, respectively~\cite{dinh2010local}. For updating $\mathbf{E}$ with~\eqref{eq:close-theta}, the complexity is $\mathcal{O}\left(KGM\right)$. As a result, the total computational complexity of \textbf{Algorithm~\ref{alg:PDD-WM}} is given by $\mathcal{O}\left(T_{\text{out}}T_{\text{inn}}\left(KGM\right)^{1.5}\right)$, where $T_{\text{out}}$ and $T_{\text{inn}}$ denote the number of outer and inner iterations required for convergence, respectively.

\subsection{Proposed Solution for WD}
Compared to WM, for WD, the baseband processing is simplified to the power allocation over waveguides. Similar to the proof in \textbf{Lemma~\ref{lemma:SINR-transform}}, problem~\eqref{eq:multicast-optimization-problem-2} for WD is equivalent to the following one:
\begin{subequations}
\label{eq:multicast-optimization-problem-WD}
\begin{equation}
\label{eq:multicast-objective-function-WD}
\max_{\left\{\mathbf{X},\mathbf{p},\gamma\right\}}\gamma,
\end{equation}
\begin{equation}
\label{eq:multicast_rate_constraint-WD}
{\rm{s.t.}} \ \ y\left(\mu_{k,g},\mathbf{X},\mathbf{p}\right)\geq \gamma, \forall k\in\mathcal{K}_{\text{MC}}, \forall g\in\mathcal{G}_{\text{k}},
\end{equation}
\begin{equation}
\label{eq:multicast_power_constraint-WD}
\sum_{k=1}^Kp_k\leq P_{\text{max}},
\end{equation}
\begin{equation}
\eqref{eq:multicast-position_constraint1}, \eqref{eq:multicast-position_constraint2}.
\end{equation}
\end{subequations}
where
\begin{align}
y&(\mu_{k,g},\mathbf{X},\mathbf{p})=2\Re\left\{\mu_{k,g}^*\sqrt{{p_k}}\mathbf{h}_{k,g}^{k}\left(\mathbf{x}_{\text{p}}^{k}\right)\mathbf{g}\left(\mathbf{x}_{\text{p}}^{k}\right)\right\}\nonumber\\&-|\mu_{k,g}|^2\left(\sum\limits_{k'\neq k}^K  p_{k'}\left|\mathbf{h}_{k,g}^{k'}\left(\mathbf{x}_{\text{p}}^{k'}\right)\mathbf{g}\left(\mathbf{x}_{\text{p}}^{k'}\right)\right|^2+\sigma_k^2\right),
\label{eq:y-mu-WD}
\end{align}
and 
\begin{align}
\label{eq:optimal_mu_WD}
\mu_{k,g}=\frac{\sqrt{{p_k}}\mathbf{h}_{k,g}^{k}\left(\mathbf{x}_{\text{p}}^{k}\right)\mathbf{g}\left(\mathbf{x}_{\text{p}}^{k}\right)}{\sum\limits_{k'\neq k}^K  p_{k'}\left|\mathbf{h}_{k,g}^{k'}\left(\mathbf{x}_{\text{p}}^{k'}\right)\mathbf{g}\left(\mathbf{x}_{\text{p}}^{k'}\right)\right|^2+\sigma_k^2}.
\end{align}
For solving problem~\eqref{eq:multicast-optimization-problem-WD}, the PDD framework can be applied again. The only change from \textbf{Algorithm~\ref{alg:PDD-WM}} is the optimization of $\mathbf{p}$ to replace that of $\mathbf{W}$. Specifically, the subproblem w.r.t. $\left\{\gamma, \mathbf{p}\right\}$ can be formulated as
\begin{subequations}
\label{eq:multicast-optimization-subproblem-WD-p}
\begin{align}
\label{eq:multicast-objective-function-WD-p}
\max_{\left\{\gamma,\mathbf{p}\right\}}\gamma,
\end{align}
\begin{equation}
\label{eq:multicast_rate_constraint-WD-p}
{\rm{s.t.}} \ \ y\left(\mu_{k,g},\mathbf{p}\right)\geq \gamma,\forall k,g,
\end{equation}
\begin{equation}
\eqref{eq:multicast_power_constraint-WD}.
\end{equation}
\end{subequations}
Problem~\eqref{eq:multicast-optimization-subproblem-WD-p} is a convex optimization problem, which can be effectively solved with CVX. By iteratively optimizing the subproblems w.r.t. $\left\{\gamma, \mathbf{p}\right\}$, $\left\{\gamma, \mathbf{U}\right\}$, $\mathbf{X}$, and $\mathbf{E}$ in the inner loop and updating the penalty terms in the outer loop in a similar manner as previously described, the MMF problem for WD can be effectively solved under the PDD framework.

\subsection{Proposed Solution for WS}
\subsubsection{Extended PDD algorithm for the multicast case}
For WS, the optimization problem~\eqref{eq:multicast-optimization-problem-2} can be addressed in two stages. We first jointly optimize $\mathbf{w}_k$ and $\mathbf{X}_k, \forall k$, which is followed by the optimization of $\boldsymbol{\lambda}$. Since the optimization of $\mathbf{w}_k$ and $\mathbf{X}_k$ can be decoupled for each group $k$ with the aim of maximizing the minimum data rate, we can obtain the following subproblem for each group $k$:
\begin{subequations}
\label{eq:multicast-optimization-problem-TS}
\begin{align}
\label{eq:multicast-objective-function-AL-TS}
\max_{\left\{\mathbf{X}_k,\mathbf{w}_k,\gamma_k\right\}}\gamma_k,
\end{align}
\begin{equation}
\label{eq:multicast_rate_constraint-AL-TS}
{\rm{s.t.}} \ \ \frac{\left|\mathbf{h}_{k,g}\left(\mathbf{X}_k\right)\mathbf{G}\left(\mathbf{X}_k\right)\mathbf{w}_k\right|^2}{\sigma_k^2}\geq \gamma_k, \forall g\in\mathcal{G}_{\text{k}},
\end{equation}
\begin{equation}
\label{eq:multicast_power_constraint-TS}
\left\|\mathbf{w}_k\right\|^2\leq P_{\text{max}},
\end{equation}
\begin{equation}
\eqref{eq:multicast-position_constraint1}, \eqref{eq:multicast-position_constraint2}.
\end{equation}
\end{subequations}
For solving problem~\eqref{eq:multicast-optimization-problem-TS}, the PDD framework can be employed again by involving the inner and outer loops similar to that in \textbf{Algorithm~\ref{alg:PDD-WM}}. Then, the stationary solutions of $\mathbf{w}_k$ and $\mathbf{X}_k$ can be obtained. For the sake of clarity, the algorithm details are omitted here.  

With the obtained $\mathbf{w}_k$ and $\mathbf{X}_k, \forall k$, problem~\eqref{eq:multicast-optimization-problem-2} for WS is reduced to the following subproblem:
\begin{subequations}
    \label{eq:time-allocation-problem}
    \begin{equation}
        \label{eq:time-allocation-objective}
        \max_{\boldsymbol{\lambda},\xi}\xi,
    \end{equation}
    \begin{equation}
        \text{s.t.}\ \ \label{eq:time-allocation-constraint}\lambda_k\log_2\left(1+\frac{\left|\mathbf{h}_{k,g}\left(\mathbf{X}_k\right)\mathbf{G}\left(\mathbf{X}_k\right)\mathbf{w}_k\right|^2}{\sigma_k^2}\right)\geq \xi, \forall k, g,
    \end{equation}
    \begin{equation}
        \eqref{eq:multicast-time-allocation-constraint}.
    \end{equation}
\end{subequations}
Problem~\eqref{eq:time-allocation-problem} is a convex optimization problem, which can be effectively solved with CVX. 
\subsubsection{Low-complexity algorithm for the unicast case}
The unicast case for WS can be considered as a single-user scenario for each given time slot, where we propose a low-complexity algorithm for solving the formulated MMF problem. Specifically, it is well known that the maximum-ratio transmission (MRT) beamformer achieves the optimal baseband beamforming performance in the single-user case. Specifically, for any given time slot for serving user $k$, the baseband MRT beamformer with given $\mathbf{X}_k$ is given by 
\begin{equation}
    \label{eq:MRT-beamformer}\mathbf{w}^*_k=\sqrt{P_{\text{max}}}\frac{\left(\mathbf{h}_{k}\left(\mathbf{X}_k\right)\mathbf{G}\left(\mathbf{X}_k\right)\right)^H}{\left\|\mathbf{h}_{k}\left(\mathbf{X}_k\right)\mathbf{G}\left(\mathbf{X}_k\right)\right\|}.
\end{equation}

Since no inter-user interference exists, the objective for designing $\mathbf{X}_k$ is to minimize the large-scale path loss and constructively combine the received signals from $N$ PAs on each waveguide at user $k$. Notably, as long as the signal phases from PAs on each waveguide are aligned, the phase differences over PAs on different waveguides can be effectively mitigated via the MRT beamformer.  
Hence, for serving user $k$, the MMF problem for optimizing $\mathbf{X}_k$ can be relaxed to the following tractable one:
\begin{subequations}
\label{eq:TS-unicast-optimization-problem}
    \begin{equation}
        \label{eq:TS-unicast-objective-function}\max_{\mathbf{X}_k}\sum\limits_{i=1}^K\sum\limits_{n=1}^N\frac{1}{\left\|\boldsymbol{\psi}_{\text{u}}^k-\boldsymbol{\psi}_{\text{p}}^{i,n}\right\|},
    \end{equation}
    \begin{equation}
    \label{eq:2pi-constraint}
        \phi_{i,k}^n - \phi_{i,k}^{n'}=2j\pi, \forall i\in\left\{1,\dots, K\right\}, n\neq n',
    \end{equation}
    \begin{equation}
        \eqref{eq:multicast-position_constraint1}, \eqref{eq:multicast-position_constraint2},
    \end{equation}
\end{subequations}
where $j$ is an arbitrary integer. The rewritten objective function~\eqref{eq:TS-unicast-objective-function} is to maximize the sum of the reciprocals of the distances from PAs to the user, so as to minimize the path loss. Constraint~\eqref{eq:2pi-constraint} guarantees that signals from all PAs on each waveguide are constructively combined at the user. 

For solving problem~\eqref{eq:TS-unicast-optimization-problem} under the single-waveguide case, i.e., $K=1$, a low-complexity two-stage algorithm was proposed in~\cite{yanqing}. Specifically, PAs positions are firstly optimized to minimize the path loss in a large-scale manner, which is followed by the small-scale refinement for phase alignment. Observing problem~\eqref{eq:TS-unicast-optimization-problem}, PAs positions over each waveguide are explicitly decoupled in both the objective function and the constraints. Hence, the two-stage optimization approach designed in~\cite{yanqing} can be directly applied to each waveguide $i\in\left\{1,\dots,K\right\}$. For more details of the two-stage algorithm, the reader is referred to~\cite[Section III]{yanqing}. Then, the remaining time allocation problem can be solved in the same way as that for problem~\eqref{eq:time-allocation-problem}. The complete algorithm is summarized in \textbf{Algorithm~\ref{alg:TS-unicast}}, for which the complexity is analyzed as follows. For the pinching beamforming, the PA-wise successive position tuning algorithm proposed in~\cite{yanqing} is with the computational complexity of $\mathcal{O}\left(M\right)$. For the baseband beamforming, the computational complexity for calculating the MRT beamformer in~\eqref{eq:MRT-beamformer} is $\mathcal{O}\left(MK\right)$. For solving the time allocation problem, the computational complexity is $\mathcal{O}\left(K^3\right)$. Therefore, the overall computational complexity of \textbf{Algorithm~\ref{alg:TS-unicast}} is $\mathcal{O}\left(MK+K^3\right)$.

\begin{algorithm}[h]
\renewcommand{\algorithmicrequire}{\textbf{Input:}}
\renewcommand{\algorithmicensure}{\textbf{Output:}}
\caption{Proposed Low-Complexity Algorithm for WS-Based Unicast Communications}
\label{alg:TS-unicast}
    \begin{algorithmic}[1]
    \Statex \underline{Pinching beamforming:}
    \State Obtain the PAs positions $\mathbf{X}_k$ for serving each user by solving the relaxed problem~\eqref{eq:TS-unicast-optimization-problem}.
    \Statex \underline{Baseband beamforming:}
    \State For given $\mathbf{X}_k$, obtain the baseband MRT beamformer $\mathbf{w}_k$ for each user based on~\eqref{eq:MRT-beamformer}.
    \Statex \underline{Time allocation:}
    \State For given $\left\{\mathbf{X}_k, \mathbf{w}_k\right\}$, obtain the optimal $\boldsymbol{\lambda}$.
    \end{algorithmic}
\end{algorithm}

\textcolor{blue}{
\subsection{Practical Overheads Discussion}
Although the above solutions have provided efficient approaches for joint baseband processing and pinching beamforming in PASS, it is worthy noting that, the practical implementation of multi-waveguide PASS also needs to consider the following practical issues. Firstly, the fabrication imperfection may cause PAs positional errors, which can affect the pinching beamforming performance. Therefore, PAs are expected to be equipped with high-precision position sensors and control circuitries for more efficient pinching beamforming. Secondly, PAs activation latency needs to be restricted within the channel coherence time. If the communication environment undergoes rapid changes, more efficient PAs activation mechanism needs to be designed. Furthermore, since WS imposes stringent time synchronization across users, the time-synchronization budget for WS needs to take the PAs activation delay into account for avoiding slot collisions.
}

\section{Simulation Results}
In this section, numerical results are provided to verify the effectiveness of the proposed multi-user downlink PASS under the proposed three transmission structures, for both unicast and multicast scenarios. 

\begin{figure}[t]
	\centering
	\includegraphics[scale=0.32]{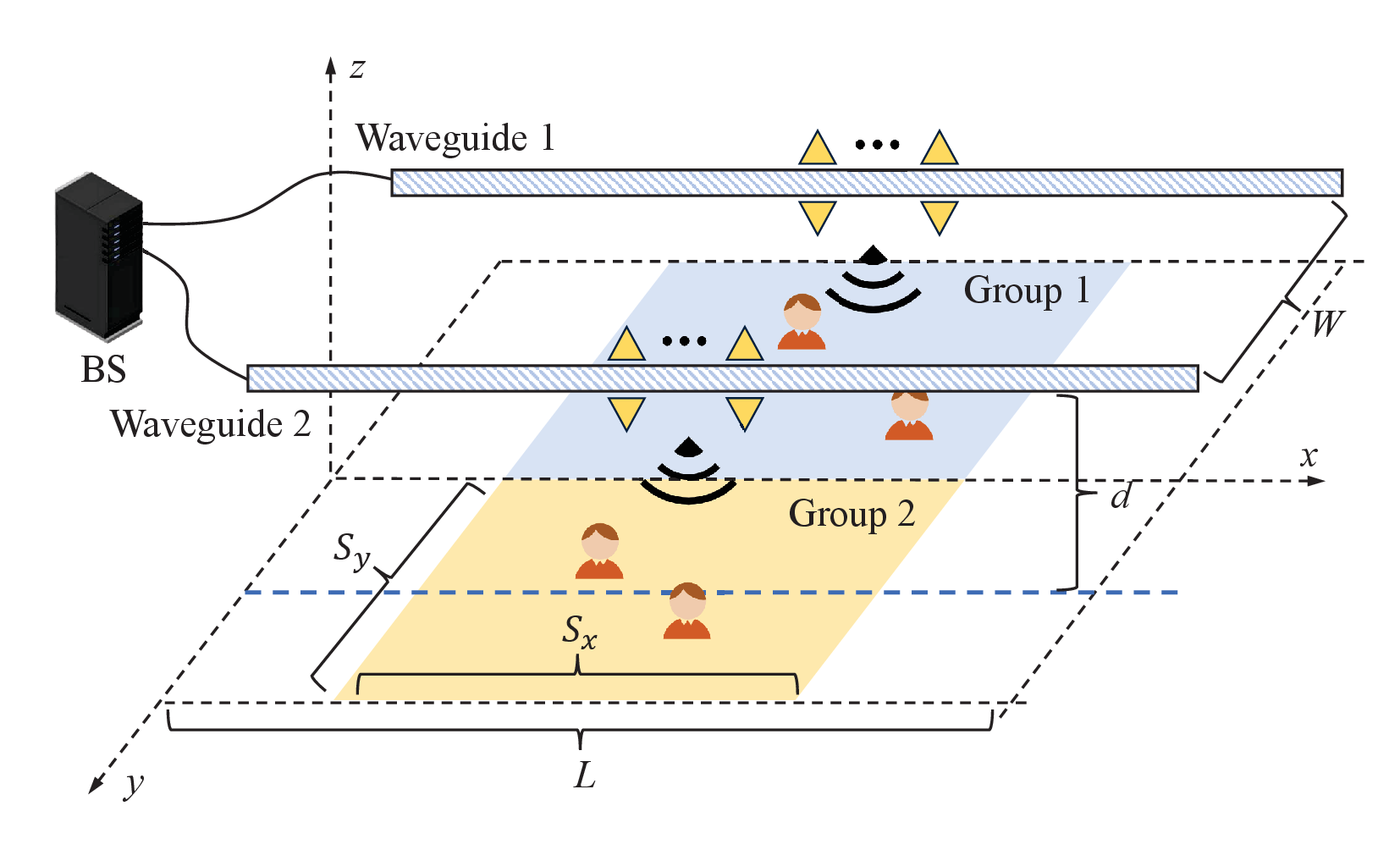}
	\caption{Simulation setup.}
	\label{fig:system_model} 
\end{figure}

\subsection{Simulation Setup}
Unless otherwise specified, we set the carrier frequency, the effective refractive index of a waveguide, the BS maximum transmit power, and the user noise power as $f_\text{c}=28$~GHz, $n_{\text{eff}}=1.4$, $P_{\text{max}}=\left\{0,5,10,15,20\right\}$~dBm, and $\sigma^2_k=-90$~dBm, respectively. 
The number of waveguides and unicast users/multicast groups is set as $K=2$, and the number of users in each group is set as $G=2$ for the multicast case. Each waveguide is connected with a dedicated RF chain and equipped with $N=\left\{4,6,8,10,12\right\}$ PAs. Regarding the geometrical configuration, consider that all waveguide with the length $L=10$~m are placed at the height of $d=3$~m in parallel, with equal distance of $W=\left\{5,10,15,20,25,30,35,40\right\}$~m, as shown in Fig.~\ref{fig:system_model}. The coordinates of the $n$-th PA on two waveguides are $\left(x_{\text{p}}^{1,n}, -\frac{W}{2}, 3\right)$ and $\left(x_{\text{p}}^{2,n}, \frac{W}{2}, 3\right)$, respectively. The minimum spacing among PAs is set as $\lambda/2$, with $\lambda=c/f_{\text{c}}$.
Define two rectangular areas $\mathcal{S}_1=\left\{[x,y,0]\left|x\in\mathcal{C}_x, y\in\left(-\frac{W+S_y}{2},-\frac{W-S_y}{2}\right)\right.\right\}$, $\mathcal{S}_2=\left\{[x,y,0]\left|x\in\mathcal{C}_x, y\in\left(\frac{W-S_y}{2},\frac{W+S_y}{2}\right)\right.\right\}$, with $\mathcal{C}_x=\left(5-\frac{S_x}{2},5+\frac{S_x}{2}\right)$, $S_x=\left\{2,4,6,8,10\right\}$, and $S_y=5$~m. For unicast communications, the two users are randomly located in the the two rectangles $\mathcal{S}_1$, $\mathcal{S}_2$, respectively. For multicast communications, the users in each multicast group are randomly located in $\mathcal{S}_1$, $\mathcal{S}_2$, respectively. Moreover, regarding parameters settings in the proposed \textbf{Algorithm~\ref{alg:PDD-WM}}, {the residual tolerance threshold and the convergence threshold for the inner loop are set as $\epsilon_1=10^{-6}$ and $\epsilon_2=10^{-3}$, respectively. The penalty factor and the dual variables are initialized as $\rho^{\left(0\right)}=10^{-3}$ and $\boldsymbol{\lambda}^{(0)}=\mathbf{0}$, respectively.}


\subsection{Baseline Schemes}
For performance comparison, we consider the conventional MIMO system with fully-digital and hybrid beamforming structures as the baseline schemes. Notably, for the sake of fair comparison under the same hardware cost, the number of RF chains for the baselines is set the same as that for the considered PASS.
\begin{itemize}
    \item \textbf{Conventional MIMO (Fully-digital beamforming)}: In this case, the centroid of a uniform linear array along the $x$-axis is positioned at $(5,0,3)$~m. The array consists of $K$ antennas spaced with $\lambda/2$. Each antenna is connected to a dedicated RF chain, i.e., the fully digital beamforming is enabled at the BS. Thus, for the unicast case, the signal received by the $k$-th user is given by $\overline{y}_k=\overline{\mathbf{h}}_k^H\sum_{i=1}^K\mathbf{w}_is_i+n_k$,
    where $\overline{\mathbf{h}}_k\in\mathbb{C}^{K\times1}$ denotes the channel vector from the antenna array to the user $k$. The similar signal expression can be obtained for the multicast case, which is omitted here for brevity. The resultant MMF problems can be solved directly using the CVX tool after introducing the auxiliary variables $\gamma$ and $\mu$ as done in Section IV-A.
    \item \textbf{Conventional MIMO (Hybrid beamforming)}: In this case, a uniform linear array is centered at $(5,0,3)$~m, where $K$ RF chains are connected to $M$ antennas in a sub-connected structure. Again, take the unicast case as an example, the signal received by the $k$-th user is given by  $\overline{\overline{y}}_k={\overline{\mathbf{h}}}_k^H\mathbf{W}_{\mathrm{RF}}\sum_{i=1}^K\mathbf{w}_is_i+n_k,$
    where $\mathbf{W}_{\mathrm{RF}}$ is the analog beamforming matrix. For solving the resultant MMF problems, we invoke the idea of minimizing the distance between the hybrid beamformer and the optimal fully-digital beamformer. For more details, the reader is referred to~\cite{10587118}.
\end{itemize}
\subsection{Convergence of the Proposed Algorithms}

\begin{figure}[t]
	\centering
	\includegraphics[scale=0.6]{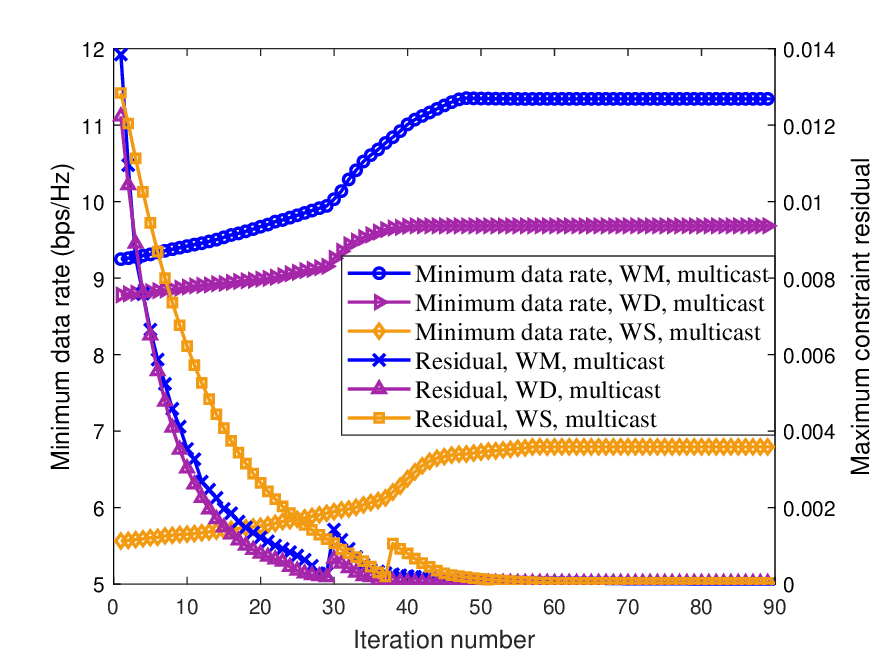}
	\caption{Convergence of the proposed PDD algorithms, with $N=8$, $W=5$~m, $S_x=6$~m, and $P_{\text{max}}=20$~dBm.}
	\label{fig:algorithm-convergence} 
\end{figure}

Fig.~\ref{fig:algorithm-convergence} demonstrates the convergence behavior of the proposed \textbf{Algorithm~\ref{alg:PDD-WM}} for WM, WD, and WS under the multicast case, where we set $N=8$, $W=5$~m, $S_x=6$~m, and $P_{\text{max}}=20$~dBm. For demonstrating the satisfaction of the equality constraints in~\eqref{eq:equality_constraint}, we also plot the maximum constraint residual, i.e., $\max\left\{\left\|\mathbf{A}^{(t)}\right\|_\infty,\left\|\mathbf{B}^{(t)}\right\|_\infty\right\}$.
The results provided are obtained from a single random realization of the users distribution. It can be observed that, the minimum data rate increases slowly at first, which is due to the fact that the penalty of constraint residuals dominates the objective functions. With the decrement of the constraint residual, the minimum data rate increases rapidly after around $30$ iterations. Moreover, for WM, WD and WS, the algorithms converge within around $50$, $40$ and $60$ iterations, respectively, when the maximum constraint residual reduces to around $0$. This implies that, the equality constraints, i.e.,~\eqref{eq:u-equality} and~\eqref{eq:e-equality}, are ultimately satisfied at convergence, which demonstrates the effectiveness of the proposed PDD algorithms. WS demonstrates the slowest convergence speed, given that more optimization variables (i.e., PAs positions tailored for each user) are involved. WD shows the most rapid convergence speed, which is attributed to the fact that the power allocation is generally more tractable and easier to handle than the baseband beamforming.


\begin{figure}
    \centering
    \begin{subfigure}{\linewidth}
        \centering
        \includegraphics[scale=0.6]{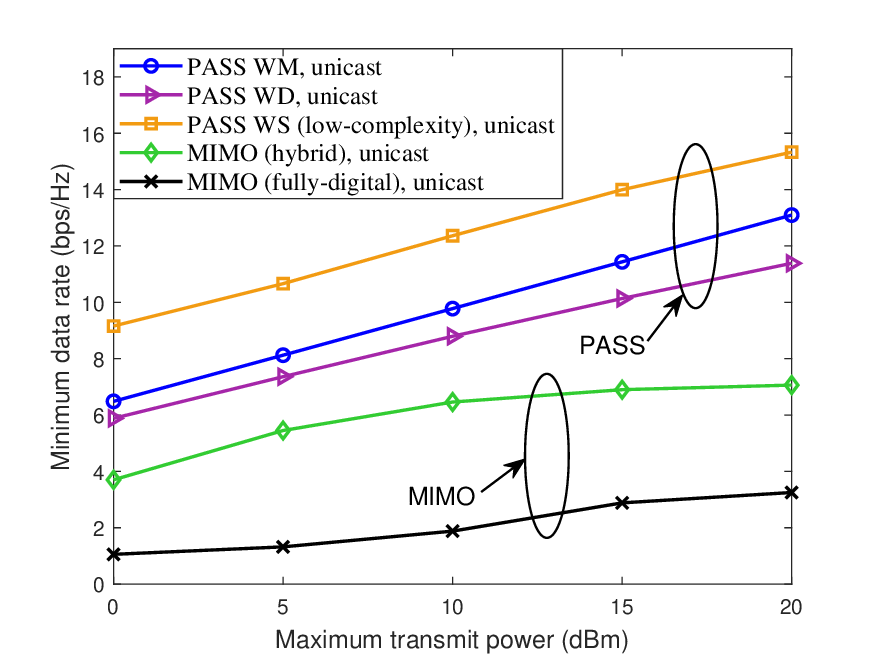}
        \caption{Unicast communications.}
        \label{fig:transmit-power-unicast}
    \end{subfigure}
    \begin{subfigure}{\linewidth}
        \centering
        \includegraphics[scale=0.6]{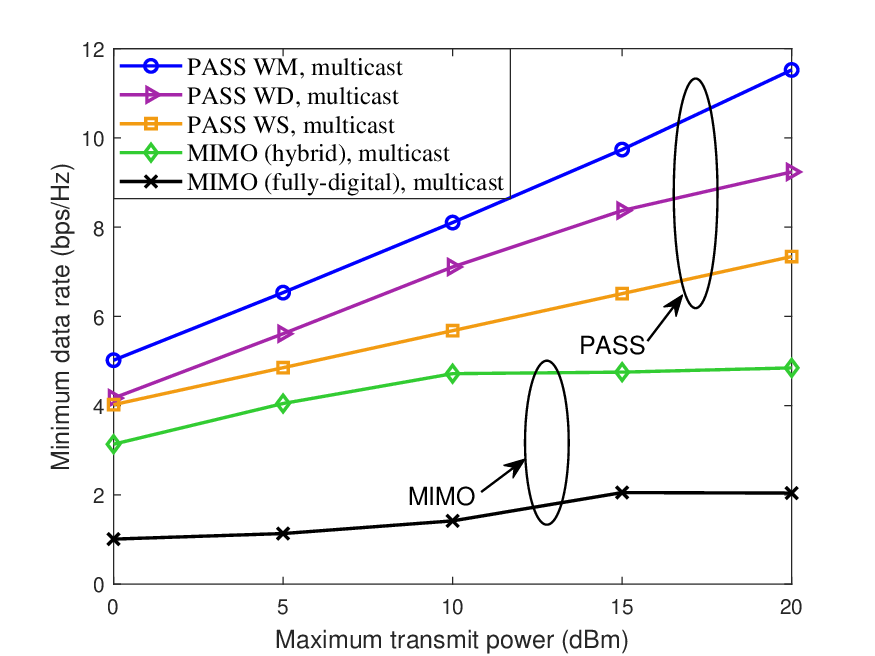}
        \caption{Multicast communications.}
        \label{fig:transmit-power-multicast}
    \end{subfigure}
    \caption{Minimum data rate versus BS maximum transmit power with $N=8$, $W=5$~m, and $S_x=6$~m.}
    \label{fig:transmit-power}
\end{figure}

\subsection{Minimum Data Rate Versus BS Maximum Transmit Power}
Fig. 4(a) and Fig. 4(b) illustrate the minimum data rate versus the BS maximum transmit power $P_{\text{max}}$ for the unicast and multicast cases, respectively, where we set $N=8$, $W=5$~m, and $S_x=6$~m. As can be observed, 
the considered PASS significantly outperforms the conventional MIMO system for both unicast and multicast communications. For instance, considering the WM structure for multicast communication with $P_{\text{max}}$ set as $20$~dBm, the minimum data rate achieved by PASS is around $11.5$~bps/Hz, while that achieved by the conventional MIMO is $2$~bps/Hz and $4.8$~bps/Hz for fully-digital and hybrid beamforming structures, respectively. This is expected, as PASS can significantly reduce the path loss by flexibly moving PAs towards users. It is also worthy noting that the performance enhancement brought by PASS is based on the premise of low-cost deployment, while the hybrid MIMO requires a large number of phase shifters for the analog beamforming. 

\textcolor{blue}{Regarding the performance comparison among the three transmission structures for PASS, as can be observed from Fig. 4(b), WS achieves the best performance for unicast communications, while WM shows its superiority for the multicast counterpart. The reason behind this can be explained as follows. For unicast communications, WS can realize the PAs positions design tailored for each individual user, which brings in significant performance gain compared to WM and WD where all users are served simultaneously.} For multicast communications, since users in the same group are served at the same time, the PAs positions optimization does not work in a user-specific manner any more. Then, WM becomes attractive as the full time resources can be fully leveraged to serve all users. Also note that WM outperforms WD for both unicast and multicast communications, as WM can fully harness the multiplexing gain with the baseband beamforming.   

\begin{figure}[t]
	\centering
	\includegraphics[scale=0.6]{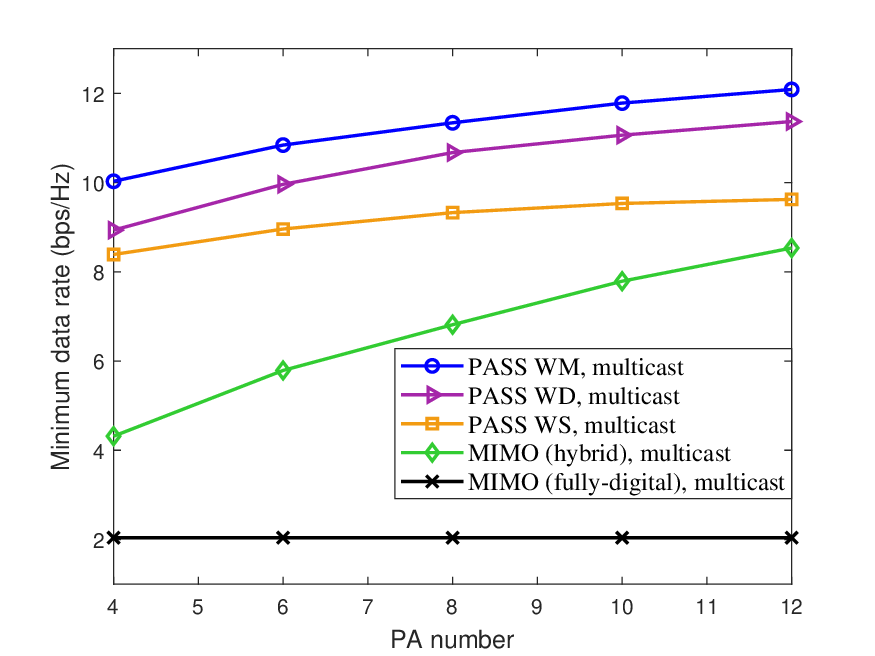}
	\caption{Minimum data rate versus the number of PA with $W=5$~m, $S_x=6$~m, and $P_{\text{max}}=20$~dBm.}
	\label{fig:PA-number} 
\end{figure}

\subsection{Minimum Data Rate Versus Number of PAs}
Fig.~\ref{fig:PA-number} depicts the minimum data rate versus the number of PAs $N$ for multicast communications, where we set $W=5$~m, $S_x=6$~m, and $P_{\text{max}}=20$~dBm. As can be seen from the figure, the minimum data rate increases with the increment of $N$. For example, when the number of PAs on each waveguide increases from $4$ to $12$, the minimum data rate is improved by $20\%$, $27\%$, and $15\%$ for WM, WD, and WS, respectively. This is due to the fact that, a larger $N$ leads to higher spatial DoFs for the pinching beamforming, thus improves the capability of strengthening the desired signal as well as mitigating inter-group interference.

\begin{figure}[t]
	\centering
	\includegraphics[scale=0.6]{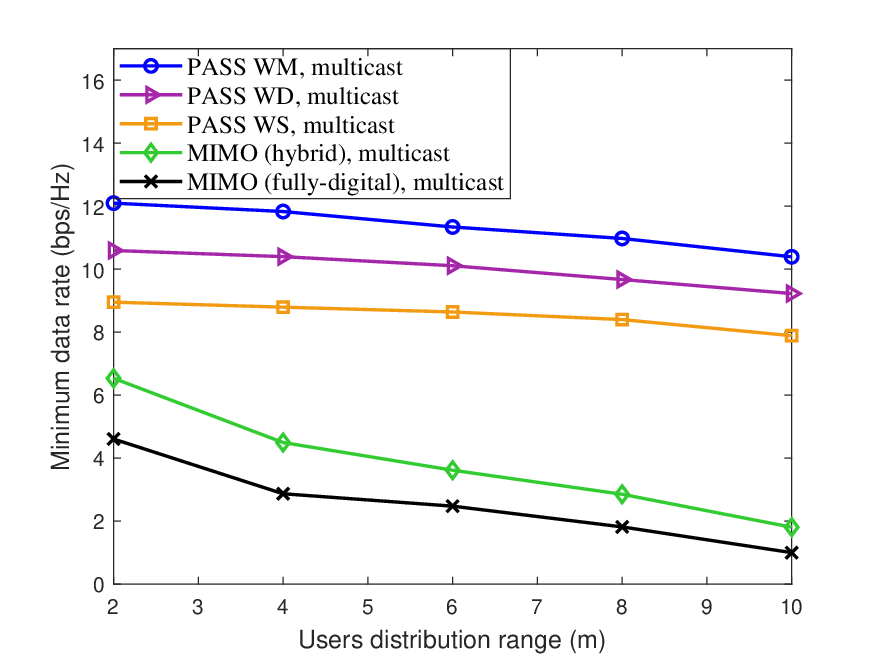}
	\caption{Minimum data rate versus users distribution range with $N=8$, $W=5$~m, and $P_{\text{max}}=20$~dBm.}
	\label{fig:usre-distribution-range} 
\end{figure}

\subsection{Minimum Data Rate Versus Users Distribution Range}
Fig.~\ref{fig:usre-distribution-range} studies the impact of the users distribution range $S_x$ on the minimum data rate for both PASS and conventional MIMO, where we set $N=8$, $W=5$~m, and $P_{\text{max}}=20$~dBm. We notice that the minimum data rate decreases dramatically with enlarged $S_x$ for the conventional MIMO. This is because, on the one hand, as antennas are fixed at the centroid of the considered users distribution area, the expanded range can lead to increased path loss, and thereby reducing the users service quality. On the other hand, when $S_x$ increases, users are more likely to be located at the endward direction of the linear array, which results in weakened array gain compared to that in the perpendicular direction. 
In contrast, the performance of PASS exhibits a slight degradation with the extension of the serving area, given that the flexible adjustment of PAs positions facilitates establishing communication links near to the users. This phenomenon underscores the superiority of PASS for long-distance transmission.

\begin{figure}[t]
	\centering
	\includegraphics[scale=0.6]{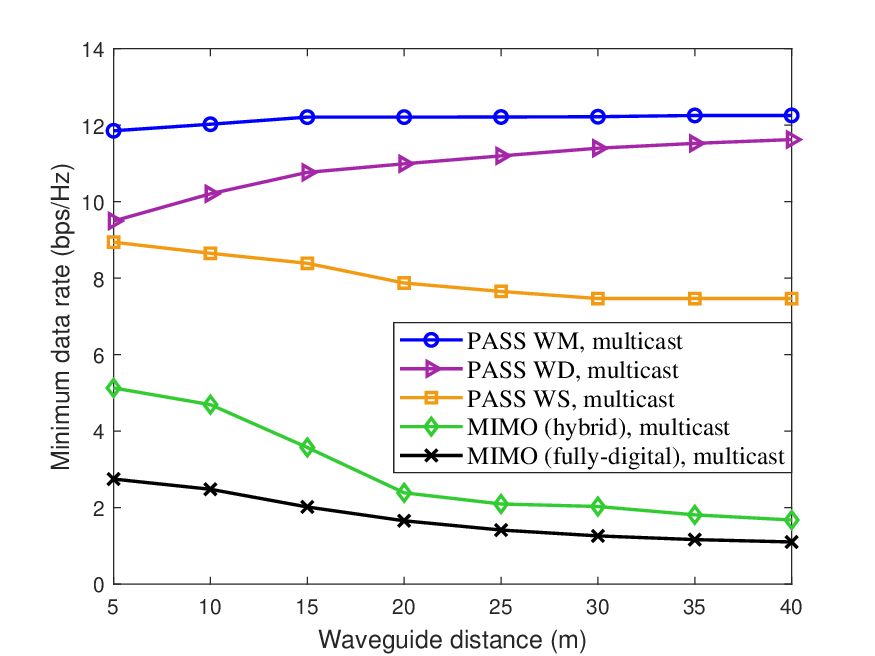}
	\caption{minimum data rate versus waveguide distance with $N=8$, $S_x=6$~m, and $P_{\text{max}}=20$~dBm.}
	\label{fig:waveguide-distance} 
\end{figure}

\subsection{Minimum Data Rate Versus Waveguide Distance}
Fig.~\ref{fig:waveguide-distance} presents the minimum data rate versus the waveguide distance $W$, where we set $N=8$, $S_x=6$~m, and $P_{\text{max}}=20$~dBm. We can first observe that the performance degrades with the larger distance for conventional MIMO system, which is expected due to the increased average path loss and the diminished beamforming gain at the axial direction. Regarding the performance of the three transmission structures for PASS, it is interesting to find that, the achievable minimum data rate of WS decreases with larger $W$, while WD shows improved performance. Meanwhile, the performance of WM remains nearly unchanged. This can be explained as follows. For WS, the signal of each user is conveyed to all waveguides with the baseband beamforming, which results in eroded data rate given the longer average distance from PAs to each user. \textcolor{blue}{For WD, since the signal of each user is only transmitted by PAs on the dedicated waveguide, the increment of $W$ contributes to reduced inter-group interference, and thus the minimum data rate is enhanced. For WM, although increasing $W$ helps mitigating inter-group interference, the pinching beamforming gain is also reduced with more widely dispersed PAs. Thus, WM shows more stable performance when $W$ gets larger. Notably, when $W$ reaches up to $40$~m, WD achieves the similar performance to WM, which implies that the baseband beamforming becomes unappealing due to the geographically separated users distribution (e.g., different waveguides are installed for serving different local areas). These observations suggest that the transmission structure for PASS needs to be carefully selected considering the practical implementation scenario.} 

\section{Conclusions}
In this paper, three practical transmission structures for multi-waveguide PASS were proposed. For each of these transmission structures, the joint baseband processing and pinching beamforming MMF optimization problem was formulated for a general multi-group multicast communication, where the unicast communication could be regarded as a special case. A PDD-based optimization framework was invoked to solve the formulated MMF problems for the three transmission structures. Numerical results demonstrated that PASS can achieve considerable MMF enhancement compared to the conventional fixed-position antenna system. Moreover, WS and WM are preferable for unicast and multicast communications, respectively. It also revealed that the performance disparity between WD and WM is substantially reduced when users were geographically separated. These insights offer useful guidance for the multi-user communications design in PASS. \textcolor{blue}{Based on the proposed transmission structures, future works can be extended to topics such as integrated sensing and communications for PASS, PASS-enabled physical layer security, etc. Due to different characteristics of the three transmission structures, their pros and cons for the application in various communication/sensing scenarios deserve further investigation.} 
\vspace{-0.3cm}
\begin{appendices}
\section{PROOF OF LEMMA 1}
\label{app:lemma-1-proof}
For any given $\mathbf{X}$ and $\mathbf{W}$, the expression $y(\mu_{k,g},\mathbf{X},\mathbf{W})$ can be reduced to $y(\mu_{k,g})$. Since our objective is to maximize the minimum achievable rate among users, $\gamma$ should be as large as possible at
the optimized solutions of problem~\eqref{eq:multicast-optimization-problem-WM2}. Therefore, constraint~\eqref{eq:multicast_rate_constraint-WM2} can be equivalently rewritten as 
\begin{equation}
\gamma \leq y\left(\mu_{k,g}\right)\triangleq y_{\max}(\mu_{g,k}), \forall k\in\mathcal{K}_{\text{MC}}, \forall g\in\mathcal{G}_{\text{k}}.
\end{equation}
It can be observed that $y\left(\mu_{k,g}\right)$ is a concave function with respect to $\mu_{g,k}$. By taking the first derivative of $y\left(\mu_{k,g}\right)$, the optimal $\mu_{g,k}$ is given by~\eqref{eq:optimal_mu}. Substituting~\eqref{eq:optimal_mu} into $y\left(\mu_{k,g}\right)$, the maximal value of $y_{\max}(\mu_{g,k})$ is given by
\begin{equation}
y_{\max}(\mu_{g,k})=\frac{\left|\mathbf{h}_{k,g}\left(\mathbf{X}\right)\mathbf{G}\left(\mathbf{X}\right)\mathbf{w}_k\right|^2}{\sum\limits_{k'\neq k}^K  \left|{\mathbf{h}_{k,g}\left(\mathbf{X}\right)\mathbf{G}\left(\mathbf{X}\right)}\mathbf{w}_{k'}\right|^2+\sigma_k^2},
\end{equation}
which recovers the SINR expression for user $(k,g)$. This completes the proof.

\end{appendices}

\bibliographystyle{IEEEtran}
\bibliography{mybib}

@STRING{IEEE_J_COML       = "{IEEE} Commun. Lett."}

@STRING{IEEE_J_WCOML       = "{IEEE} Wireless Commun. Lett."}

@STRING{IEEE_J_JSAC       = "{IEEE} J. Sel. Areas Commun."}

@STRING{IEEE_J_COM        = "{IEEE} Trans. Commun."}

@STRING{IEEE_J_SP         = "{IEEE} Trans. Signal Process."}

@STRING{IEEE_J_WCOM       = "{IEEE} Trans. Wireless Commun."}

@STRING{IEEE_M_COM        = "{IEEE} Commun. Mag."}

@STRING{IEEE_WM_COM        = "{IEEE} Wireless Commun."}

@ARTICLE{10318061,
  author={Zhu, Lipeng and Ma, Wenyan and Zhang, Rui},
  journal=IEEE_J_WCOM, 
  title={Modeling and Performance Analysis for Movable Antenna Enabled Wireless Communications}, 
  year={2024},
  volume={23},
  number={6},
  pages={6234-6250},
  month={Jun.}}

@ARTICLE{8910627,
  author={Wu, Qingqing and Zhang, Rui},
  journal=IEEE_M_COM, 
  title={Towards Smart and Reconfigurable Environment: Intelligent Reflecting Surface Aided Wireless Network}, 
  year={2020},
  volume={58},
  number={1},
  pages={106-112},
  month={Jan.},
  doi={10.1109/MCOM.001.1900107}}

@book{boyd2004convex,
  title={Convex optimization},
  author={Boyd, Stephen and Vandenberghe, Lieven},
  year={2004},
  publisher={Cambridge, U.K.: Cambridge Univ. Press}
}

@ARTICLE{9264694,
  author={Wong, Kai-Kit and Shojaeifard, Arman and Tong, Kin-Fai and Zhang, Yangyang},
  journal=IEEE_J_WCOM, 
  title={Fluid Antenna Systems}, 
  year={2021},
  volume={20},
  number={3},
  pages={1950-1962},
  month={Mar.}}

@ARTICLE{6736761,
  author={Larsson, Erik G. and Edfors, Ove and Tufvesson, Fredrik and Marzetta, Thomas L.},
  journal=IEEE_M_COM, 
  title={Massive {MIMO} for next generation wireless systems}, 
  year={2014},
  volume={52},
  number={2},
  pages={186-195},
  month={Feb.}}

@ARTICLE{6375940,
  author={Rusek, Fredrik and Persson, Daniel and Lau, Buon Kiong and Larsson, Erik G. and Marzetta, Thomas L. and Edfors, Ove and Tufvesson, Fredrik},
  journal={{IEEE} Signal Process. Mag.}, 
  title={Scaling Up {MIMO}: Opportunities and Challenges with Very Large Arrays}, 
  year={2013},
  volume={30},
  number={1},
  pages={40-60},
  month={Jan.}}

@article{yanqing,
    author = {Xu, Yanqing and Ding, Zhiguo and Karagiannidis, George K.},
    title = {Rate Maximization for downlink pinching-antenna systems},
    journal = IEEE_J_WCOML,
    year = {2025},
    volume={14},
  number={5},
  pages={1431-1435},
month={May}
}

@article{ding,
    author = {Ding, Zhiguo and Schober, Robert and Poor, H. Vincent},
    title = {Flexible-antenna systems: A pinching-antenna perspective},
    journal = IEEE_J_COM,
    year = {2025},
volume={},
  number={},
  pages={1-1},
  month={Early Access,}
}

@article{docomo,
    author = {Fukuda, A and Yamamoto, H and Okazaki, H and Suzuki, Y and Kawai, K},
    title = {Pinching antenna: using a dielectric waveguide as an antenna},
    journal = {NTT DOCOMO Techinical J.},
    volume={23},
    number={3},
    pages={5-12},
    year = {2022},
    month={Jan.}
}

@ARTICLE{11169486,
  author={Liu, Yuanwei and Wang, Zhaolin and Mu, Xidong and Ouyang, Chongjun and Xu, Xiaoxia and Ding, Zhiguo},
  journal=IEEE_M_COM, 
  title={Pinching-Antenna Systems: Architecture Designs, Opportunities, and Outlook}, 
  year={Early Access, 2025. doi:10.1109/MCOM.001.2500037},
  volume={},
  number={},
  pages={}
}

@article{kaidi,
    author = {Wang, K and Ding, Z and Schober, R},
    title = {Antenna activation for {NOMA} assisted pinching-antenna systems},
    journal = IEEE_J_WCOML,
    year={2025},
  volume={14},
  number={5},
  pages={1526-1530},
month={May}
    
}

@book{microwave,
    author = {Pozar, D M.},
    title = {Microwave engineering: theory and techniques},
    publisher = {John wiley \& sons},
    year = {2021}
}

@ARTICLE{marco,
  author={Di Renzo, Marco and Zappone, Alessio and Debbah, Merouane and Alouini, Mohamed-Slim and Yuen, Chau and de Rosny, Julien and Tretyakov, Sergei},
  journal=IEEE_J_JSAC, 
  title={Smart Radio Environments Empowered by Reconfigurable Intelligent Surfaces: How It Works, State of Research, and The Road Ahead}, 
  year={2020},
  volume={38},
  number={11},
  pages={2450-2525},
  month={Nov.}}

@ARTICLE{9324910,
  author={Shlezinger, Nir and Alexandropoulos, George C. and Imani, Mohammadreza F. and Eldar, Yonina C. and Smith, David R.},
  journal=IEEE_WM_COM, 
  title={Dynamic Metasurface Antennas for {6G} Extreme Massive {MIMO} Communications}, 
  year={2021},
  volume={28},
  number={2},
  pages={106-113},
  month={Apr.}}

@ARTICLE{9136592,
  author={Huang, Chongwen and Hu, Sha and Alexandropoulos, George C. and Zappone, Alessio and Yuen, Chau and Zhang, Rui and Renzo, Marco Di and Debbah, Merouane},
  journal=IEEE_WM_COM, 
  title={Holographic {MIMO} Surfaces for {6G} Wireless Networks: Opportunities, Challenges, and Trends}, 
  year={2020},
  volume={27},
  number={5},
  pages={118-125},
  month={Oct.}}

@ARTICLE{10909665,
  author={Tegos, Sotiris A. and Diamantoulakis, Panagiotis D. and Ding, Zhiguo and Karagiannidis, George K.},
  journal=IEEE_J_WCOML, 
  title={Minimum Data Rate Maximization for Uplink Pinching-Antenna Systems}, 
  year={2025},
  volume={14},
  number={5},
  pages={1516-1520},
  month={May}
}

@ARTICLE{10976621,
  author={Tyrovolas, Dimitrios and Tegos, Sotiris A. and Diamantoulakis, Panagiotis D. and Ioannidis, Sotiris and Liaskos, Christos K. and Karagiannidis, George K.},
  journal={IEEE Trans. Cogn. Commun. Netw.}, 
  title={Performance Analysis of Pinching-Antenna Systems}, 
  year={Early Access, 2025},
  volume={},
  number={},
  pages={}
}

@ARTICLE{10587118,
  author={Wang, Zhaolin and Mu, Xidong and Liu, Yuanwei},
  journal=IEEE_J_COM, 
  title={Beamfocusing Optimization for Near-Field Wideband Multi-User Communications}, 
  year={2025},
  volume={73},
  number={1},
  pages={555-572},
  doi={10.1109/TCOMM.2024.3424244},
  month={Jul.}}

@ARTICLE{9120361,
  author={Shi, Qingjiang and Hong, Mingyi},
  journal=IEEE_J_SP, 
  title={Penalty Dual Decomposition Method for Nonsmooth Nonconvex Optimization—Part {I}: Algorithms and Convergence Analysis}, 
  year={2020},
  volume={68},
  number={},
  pages={4108-4122},
  doi={10.1109/TSP.2020.3001906},
  month={Jun.}}

@ARTICLE{9119203,
  author={Shi, Qingjiang and Hong, Mingyi and Fu, Xiao and Chang, Tsung-Hui},
  journal=IEEE_J_SP, 
  title={Penalty Dual Decomposition Method for Nonsmooth Nonconvex Optimization—Part {II}: Applications}, 
  year={2020},
  volume={68},
  number={},
  pages={4242-4257},
  doi={10.1109/TSP.2020.3001397},
  month={Jun.}}

@article{xu2013block,
  title={A block coordinate descent method for regularized multiconvex optimization with applications to nonnegative tensor factorization and completion},
  author={Xu, Yangyang and Yin, Wotao},
  journal={J. Imaging Sci.},
  volume={6},
  number={3},
  pages={1758--1789},
  year={2013},
  publisher={SIAM},
  month={Jan.}
}

@inproceedings{dinh2010local,
  title={Local convergence of sequential convex programming for nonconvex optimization},
  author={Dinh, Quoc Tran and Diehl, Moritz},
  booktitle={Recent Advances in Optimization and its Applications in Engineering},
  pages={},
  year={2010},
  organization={Berlin, Germany: Springer}
}

@ARTICLE{10981775,
  author={Ouyang, Chongjun and Wang, Zhaolin and Liu, Yuanwei and Ding, Zhiguo},
  journal=IEEE_J_COML, 
  title={Array Gain for Pinching-Antenna Systems {(PASS)}}, 
  year={2025},
  volume={29},
  number={6},
  pages={1471-1475},
  doi={10.1109/LCOMM.2025.3566299},
  month={Jun.}}

@ARTICLE{11018390,
  author={Xiao, Jian and Wang, Ji and Liu, Yuanwei},
  journal=IEEE_J_COML, 
  title={Channel Estimation for Pinching-Antenna Systems ({PASS})}, 
  year={2025},
  volume={29},
  number={8},
  pages={1789-1793},
  doi={10.1109/LCOMM.2025.3575166},
  month={Aug.}}

@article{shan2025multigroup,
  title={Multigroup Multicast Design for Pinching-Antenna Systems: Waveguide-Division or Waveguide-Multiplexing?},
  author={Shan, Shan and Ouyang, Chongjun and Li, Yong and Liu, Yuanwei},
  journal={arXiv preprint arXiv:2506.16184},
  year={2025}
}

@article{xu2025pinching,
  title={Pinching-Antenna Systems with In-Waveguide Attenuation: Performance Analysis and Algorithm Design},
  author={Xu, Yanqing and Ding, Zhiguo and Schober, Robert and Chang, Tsung-Hui},
  journal={arXiv preprint arXiv:2506.23966},
  year={2025}
}

@ARTICLE{11180028,
  author={Oikonomou, Thrassos K. and Tegos, Sotiris A. and Diamantoulakis, Panagiotis D. and Liu, Yuanwei and Karagiannidis, George K.},
  journal=IEEE_J_COML, 
  title={{OFDMA} for Pinching-Antenna Systems}, 
  year={Early Access, 2025. doi:10.1109/LCOMM.2025.3614552},
  volume={},
  number={},
  pages={},
}

@article{liu2025pinching,
  title={Pinching-antenna systems ({PASS}): A tutorial},
  author={Liu, Yuanwei and Jiang, Hao and Xu, Xiaoxia and Wang, Zhaolin and Guo, Jia and Ouyang, Chongjun and Mu, Xidong and Ding, Zhiguo and Nallanathan, Arumugam and Karagiannidis, George K and others},
  journal={arXiv preprint arXiv:2508.07572},
  year={2025}
}

@ARTICLE{11159291,
  author={Zhao, Jingjing and Xue, Songtao and Cai, Kaiquan and Mu, Xidong and Liu, Yuanwei and Zhu, Yanbo},
  journal=IEEE_J_JSAC, 
  title={Near-Field Integrated Sensing and Communications for Secure {UAV} Networks}, 
  year={Early Access, 2025. doi:10.1109/JSAC.2025.3608737},
  volume={},
  number={},
  pages={},
}
\end{document}